\documentclass[a4paper]{article}
\usepackage{amsmath}
\usepackage{amsthm}
\usepackage{comment}
\usepackage{graphicx} 
\usepackage[utf8]{inputenc}
\usepackage[capitalise]{cleveref}
\usepackage{xcolor}
\usepackage[subrefformat=parens]{subcaption}
\usepackage{url}
\usepackage{multirow}
\usepackage{nccmath} 

\usepackage{bm}
\usepackage{natbib} 
\usepackage[margin=25truemm]{geometry}

\usepackage{amssymb}
\newtheorem{proposition}{Proposition}
\newtheorem{lemma}{Lemma}

\title{Hypercongestion, autonomous vehicle, and \\ urban spatial structure}  
\author{Takao Dantsuji\thanks{Monash Institute of Transport Studies, Monash University, 23 College Walk,
Clayton, Victoria 3800, Australia Takao.Dantsuji@monash.edu} \thanks{Institute of Science and Engineering, Kanazawa University, Kakuma-machi, Kanazawa, Ishikawa 920-1192, Japan} \and Yuki Takayama \thanks{Department of Civil and Environmental Engineering, Tokyo Institute of Technology, 2-12-1 W6-9, Ookayama, Meguro, Tokyo 152-8550, Japan, takayama.y.af@m.titech.ac.jp}} 
\date{}

\begin{document}
\maketitle

\begin{abstract}
This paper examines the effects of hypercongestion mitigation by perimeter control and the introduction of autonomous vehicles on the spatial structures of cities. By incorporating a bathtub model, we develop a land use model where hypercongestion occurs in the downtown area and interacts with land use.  We show that hypercongestion mitigation by perimeter control decreases the commuting cost in the short-run and results in a less dense urban spatial structure in the long-run. Furthermore, we reveal that the impact of autonomous vehicles depends on the presence of hypercongestion. Introduction of autonomous vehicles may increase the commuting cost in the presence of hypercongestion and cause a decrease in suburban population, but make cities spatially expanded outward.This result contradicts that of the standard bottleneck model. When perimeter control is implemented, the introduction of autonomous vehicles decreases the commuting cost and results in a less dense urban spatial structure. These results show that hypercongestion is a key factor that can change urban spatial structures. 

    {\flushleft{{\bf Keywords:} hypercongestion; perimeter control; autonomous vehicle; bathtub model; land use model.}}
\end{abstract}

\section{Introduction}
\subsection{Background}
Traffic demand is highly concentrated during rush hours in urban cities, which causes hypercongestion (the downward--sloping part of the inverted U--shaped relationship between traffic flow and traffic density), and dispersing traffic demand can change the spatial distribution of residents in the long-run. Recent research has shown that downtown areas experience network-wide hypercongestion \cite[e.g.,][]{daganzo2007urban,geroliminis2008existence} and that  the temporal concentration of traffic demand leads to network {\it capacity drop} \cite[e.g.,][]{small2003hypercongestion,geroliminis2009cordon, arnott2013bathtub}. Capacity drop is phenomenon where  throughput (e.g., flow or outflow) decreases with an increase in the number of vehicles circulating in networks.

A number of studies have proposed traffic demand management (TDM) strategies, such as congestion pricing \cite[e.g.,][]{geroliminis2009cordon} and perimeter control \cite[e.g.,][]{haddad2012stability, tsekeris2013city}, to  disperse traffic demand during peak periods for hypercongestion mitigation as capacity drop causes the inefficient use of transportation systems. However, most of them focused on problems under the assumption that commuters do not relocate (i.e., fixed origin--destination demand). That is, they ignored changes in residents' spatial distribution in response to their commuting behavior changes and instead focused on the short-run effects of  TDM strategies. Since commuting and land use patterns influence each other substantially  \cite[][]{kanemoto1980theories}, the long-run impacts of TDM strategies on land use should be understood to promote efficient and sustainable urban developments. 

Models of urban spatial structures can describe the interaction between commuting and residential locations \cite[][]{alonso1964historic}. Traditional models that employ static congestion models \cite[e.g.,][]{anas1998urban} have been extended to incorporate  dynamic bottleneck congestion \cite[e.g., ][]{arnott1998congestion, gubins2014dynamic, takayama2017bottleneck, fosgerau2018vickrey,xu2018pareto,fosgerau2019commuting, takayama2020gains, lianalytical}. These extended models highlight the significance of the temporal distribution of traffic demand and dynamic congestion phenomena in  long-run equilibrium; however, they cannot incorporate hypercongestion, where capacity can drop over time, unlike in bottleneck congestion. A suitable model is yet to be developed to examine the long-run effectiveness of TDM strategies for hypercongestion mitigation. Furthermore, whether hypercongestion is a key factor that can alter urban spatial structures remains an open question in the literature.

In this new era of autonomous vehicles, hypercongestion  can play a more important role in TDM strategies. According to \cite{van2016autonomous}, autonomous vehicle can enhance network capacity because they safely drive closer to each other than vehicles driven by humans  (``Network capacity" effect hereafter), and passengers in autonomous vehicles are less concerned  about travel times because they do not need to drive anymore (``VOT effect" hereafter). \cite{van2016autonomous} showed that both network capacity and VOT effects decrease the commuting cost in fully automated traffic bottleneck congestion. Furthermore, existing works \cite[e.g.,][]{lu2020impact, moshahedi2022macroscopic} demonstrated the network flow can be enhanced by the network capacity effect. However, the VOT effect on network efficiency  has been insufficiently studied because it is not as simple as the network capacity effect. The VOT effect reduces the the cost of travel time; however, an increase in the VOT effect is expected to worsen capacity drop due to the temporal concentration of traffic demand.  Thus, traffic demand patterns with autonomous vehicles will become increasingly complex in the presence of hypercongestion.  A model that systematically analyzes these impacts on the temporal and spatial distributions of traffic demand should be developed to  investigate the long-run effects of TDM strategies in this new era of autonomous vehicles properly.

In this paper, by incorporating a bathtub model, we develop a land use model where hypercongestion occurs in the downtown area and interacts with land use. Our findings show that the implementation of perimeter control for hypercongestion mitigation decreases the commuting cost, and results in a less dense urban spatial structure. Furthermore,  we find that the use of autonomous vehicles may increase the commuting cost due to the severe capacity drop caused by the temporal concentration of traffic demand, and cause an increase in downtown population, but make cities spatially expanded outward. This result contradicts that of the standard bottleneck model. When perimeter control is implemented, the introduction of autonomous vehicles decreases the commuting cost, and results in a less dense urban spatial structure. These results show that hypercongestion is a key factor that can change urban spatial structures. 

\subsection{Literature Review}
Network-wide hypercongestion, namely Macroscopic Fundamental Diagrams (MFDs) is a powerful tool for describing  network-wide traffic dynamics; this approach relates network flow (or trip completion rate) to network density (or accumulation of vehicles). The idea of macroscopic traffic theory was proposed by \cite{godfrey1969mechanism}, further investigated by \cite{daganzo2007urban}, and  empirically analyzed by \cite{geroliminis2008existence}.  Traffic management approaches based on MFDs have been studied, such as pricing \cite[e.g.,][]{zheng2012dynamic, simoni2015marginal,dantsuji2021simulation, genser2022dynamic}, perimeter control \cite[e.g.,][]{daganzo2007urban, tsekeris2013city,haddad2014robust}, route guidance \cite[e.g.,][]{yildirimoglu2015equilibrium, yildirimoglu2018hierarchical}, and road space allocation \cite[e.g., ][]{zheng2013distribution, chiabaut2015evaluation, zheng2017macroscopic}. MFDs have also been utilized for other purposes, such as  dynamic traffic demand estimation \cite[e.g., ][]{dantsuji2022novel} and network performance indication \cite[e.g.,][]{loder2019understanding}. 

Perimeter control is a successful application of MFDs in TDM strategies as an effective and easy-to-implement tool. It aims to control the entry flow at the perimeter boundary of a target area to maximize the trip completion rate. This approach has been extended to  multiregion networks \cite[e.g., ][]{geroliminis2012optimal, haddad2012stability, ramezani2015dynamics,haddad2020adaptive, sirmatel2021modeling, fu2021perimeter, batista2021role, zhou2023scalable, chen2022data, kouvelas2023linear}, perimeter control with boundary queues \cite[e.g., ][]{haddad2017optimal, ni2020city, guo2020macroscopic, li2021perimeter} and  route guidance \cite[e.g., ][]{sirmatel2017economic, ding2017traffic, fu2021perimeter}, and  bimodal transportation systems \cite[e.g.,][]{ampountolas2017macroscopic, haitao2019providing, chen2022passenger, dantsuji2022perimeter}. Considerable effort has been dedicated to perimeter control schemes, but there are no studies that examined the impacts of autonomous vehicles (i.e., VOT and network capacity effects) on their effectiveness. Furthermore, despite the significant influence of perimeter control on traffic demand, most of the studies made the assumption that origin-destination demand is fixed. This hinders a comprehensive understanding of the interplay between urban and transportation developments.   A critical gap is the lack of a methodology that connects short-run commuters' decisions (e.g., trip timing)  with their residential location choices in the long-run.

Traditional land use models that consider the trade-off between land rent and commuting costs in a monocentric city effectively analyze residential location patterns. Static congestion models \cite[e.g., ][]{kanemoto1980theories, anas1998urban} have been extended to dynamic bottleneck models \citep{arnott1998congestion} to incorporate the impact of the trip timing decisions of commuters into their residential location choices in the long-run. Several aspects have been incorporated into the standard bottleneck model as extensions, such as an incentive for commuters to spend time at home \citep{gubins2014dynamic}, bimodal transportation systems \citep{xu2018pareto},  an open city model \citep{fosgerau2018vickrey}, a tandem bottlenecks \citep{fosgerau2019commuting}, a bottleneck with a stochastic location \citep{lianalytical},  and commuters' heterogeneity  \citep{takayama2017bottleneck, takayama2020gains}.  These aspects can alter urban spatial structures in the long-run, but whether hypercongestion has the same effect remains an open question in the literature. 

The trip timing decisions of commuters in the presence of hypercongestion can be described using bathtub models,  namely dynamic user equilibrium models for departure time choices in urban cities with hypercongestion \cite[e.g.,][]{small2003hypercongestion, geroliminis2009cordon, arnott2013bathtub, fosgerau2013hypercongestion, vickrey2020congestion, bao2021leaving, ameli2022departure}. Bathtub models have been extended to trip length heterogeneity  \cite[e.g.,][]{fosgerau2013hypercongestion, lamotte2018morning}, cruising-for-parking \cite[e.g.,][]{geroliminis2015cruising, liu2016modeling}, and bimodal transportation systems \cite[e.g.,][]{gonzales2012morning, gonzales2015coordinated, dantsuji2022perimeter}. However, none of them considered  changes in residents' spatial distribution in response to their commuting behavior changes.

 Our recent work \citep{dantsuji2022perimeter} developed the bimodal bathtub model, and investigated the travelers’ behavior changes in response to perimeter control with transit priority. However, they focused only on the short-run equilibrium. How to combine the short-run equilibrium in the presence of hypercongestion with the long-run equilibrium remains a challenging topic. Also, \cite{dantsuji2022perimeter} did not consider the impacts of autonomous vehicles. While there are a vast of literature on the network capacity and VOT effects of autonomous vehicles, no one proves that the VOT effect may cause higher commuting cost in the presence of hypercongestion. To the best of our knowledge, our study is the first to systematically analyze a model in which commuters choose their trip timing in the presence of hypercongestion in the short-run and residential location in the long-run. 

\vspace{12pt}

The remaining of this paper is organized as follows. Section 2 shows the development of the model. Section 3, we characterize the model equilibrium. Section 4 presents the formulation of the model during perimeter control. The equilibrium conditions during perimeter control are studied in Section 5. Numerical examples are provided in Section 6, and we conclude this paper in Section 7.

\section{Model}
Consider a monocentric city that has downtown and suburban areas. All job opportunities are found and homogeneously distributed in the downtown area, whereas there are residential areas in both downtown and suburban areas.  The downtown area is in the center of the city, and a residential location in the suburban area is indexed by a distance $x$ from the edge of the downtown area (Fig~\ref{fig:model}). The areas of downtown and suburban at location $x$ have  fixed land areas of $A_d$ and $A_s(x)$, respectively. We assume that the land is owned by absentee landlords, as  is common in the literature. An $N$ continuum of homogeneous commuters  have identical preferences and the numbers of downtown commuters and suburban commuters at location $x$ are denoted by $N_d$ and $N_s(x)$, respectively.

\begin{figure}[t]
\centering
\includegraphics[width=0.7\textwidth]{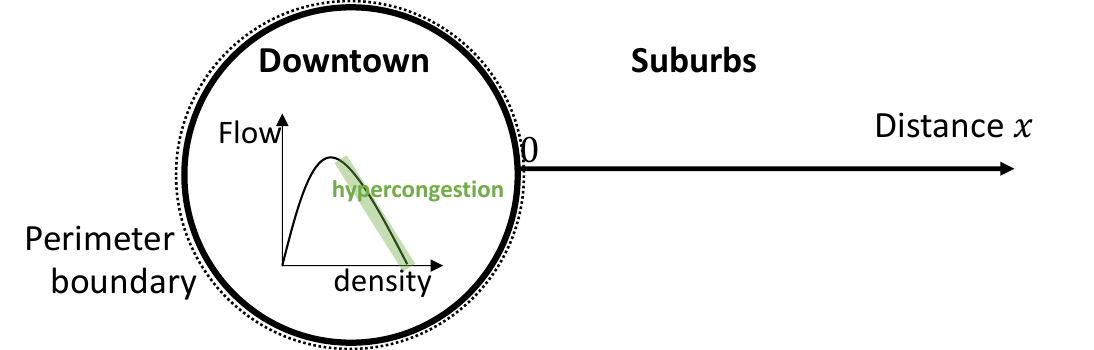}
\caption{Model structure}
\label{fig:model}
\end{figure}

\subsection{Hypercongestion and commuting cost}

The downtown area has homogeneous topological characteristics by proper partitioning approaches \cite[e.g., ][]{ji2012spatial, dantsuji2019cross},  thus showing a well-defined
relationship between space-mean flow and density. Note that as we assume that the job opportunities are homogeneously distributed in the downtown area, the shape of the MFD does not change even when commuters' residential locations change. The congestion dynamics in the downtown area are described using a bathtub model, whereas we assume that one can travel at a free-flow speed in the suburban area.  To incorporate hypercongestion, we employ the Greenshields model for the space-mean speed  as follows. 
\begin{align}\label{eq:speed}
 v(t) = v_f \left( 1 - \frac{n(t)}{n_j} \right),
\end{align}
where $v(t)$ is the space-mean speed at time $t$, $v_f$ is the free-flow speed, $n(t)$ is the vehicle accumulation in the downtown area at time $t$ and $n_j$ is jam accumulation. 

Since the downtown area is modeled as a system with inflows and outflows and whose traffic conditions are governed by bathtub congestion dynamics, the time evolution of vehicle accumulation, $\dot{n}(t)$, is 
\begin{align}\label{eq:accumulation_evolution}
\dot{n}(t) = I (t) - G(t),
\end{align}
where $I(t)$ is the inflow and $G(t)$ is the outflow at time $t$. The outflow is formulated using the network exit function (NEF) as follows:
\begin{align}\label{eq:NEF}
G(t) = \frac{n(t) v(t)}{L},
\end{align}
where $L$ is the average trip length in the downtown area.

The travel time in the downtown area is assumed to be determined by a single instant of time for  tractability  \citep{small2003hypercongestion, geroliminis2009cordon} and calculated as
\begin{align}\label{eq:TT_assumption}
T(t) \approx \frac{L}{v(t)}.
\end{align}
All suburban commuters have identical trip lengths $L$ in the downtown area.

This instantaneous travel time assumption is based on the so-called accumulation-based modelling, where the outflow is a function of the instantaneous accumulation \citep{mariotte2017macroscopic}. According to \cite{mariotte2017macroscopic}, in cases where the inflow changes drastically over time, information propagates too fast, which leads to a physical inconsistency. An alternative approach is the trip-based model \cite[e.g., ][]{fosgerau2015congestion, mariotte2017macroscopic, lamotte2018morning,  jin2020generalized}, which can tackle this inconsistency. However, the trip-based model is typically intractable. Even though the inconsistency occurs at the beginning and end of the rush hour when the inflow changes drastically, it does not affect the qualitative results in this paper. Therefore, we employ the accumulation-based model for tractability.

We assume that the downtown commuters  commute by walking or cycling and the suburban commuters travel by car, which is similar to the modeling of \cite{arnott1998congestion}. The commuting cost of a downtown commuter and a suburban commuter at location $x$ who arrive at work at time $t$ ($C_d(t)$ and $C_s(x, t)$, respectively) are expressed as
\begin{subequations}\label{eq:travel_cost}
\begin{align}
&C_d(t) = \alpha T_d + s(t),  \label{eq:commuting_cost_downtown} \\
&C_s(x, t) =  \alpha \left( T(t) + \tau x \right) + s(t), \label{eq:commuting_cost_suburban}  \\
& s(t) = 
\begin{cases}
\beta \left( t^* - t \right) & \qquad  \text{if } t \leq t^*  \\
\gamma \left( t -  t^* \right) & \qquad \text{if }t > t^* 
\end{cases}, \label{eq:schedule_delay}
\end{align} 
\end{subequations}
where $T_d$ represents the constant travel time of the downtown commuters, $T(t)$ represents the travel time in the downtown area at time $t$, $\tau x$ represents the suburban free-flow travel time of commuters residing at $x$ (i.e., $\tau = 1 / v_f$),  $t^*$ represents the desired arrival time, and $\alpha$, $\beta$ and $\gamma$ are marginal costs of travel time, earliness, and lateness. The first and second terms on the RHS of Eq.~(\ref{eq:commuting_cost_suburban}) are the travel time cost and schedule delay cost, respectively. That is, we assume that commuters have ``$\alpha$-$\beta$-$\gamma$'' type preference \citep{arnott1993structural}.

We also consider a situation where every vehicle is an autonomous vehicle. Autonomous vehicles are expected to have two effects: the VOT effect and the network capacity effect \citep{van2016autonomous}.  In this study, the former effect is represented by a reduction in the VOT. It is $\eta \alpha$ where $\eta$ is the VOT effect parameter  ($\frac{\beta}{\alpha}<\eta \leq 1$). The latter effect is captured by an increase in network capacity, which is $\xi n_j$ where  $\xi$ is the network capacity effect parameter ($\xi\geq1$). This effect increases not only the network capacity but also the critical accumulation, which is consistent with the simulation analysis by \cite{lu2020impact}.

There is a large literature on the VOT effect since it can be estimated by stated preference surveys. \cite{steck2018autonomous} found privately owned autonomous vehicle reduces value of travel time savings (VTTS) by 31 \%. \cite{zhong2020will} investigated the VOT effects in suburban, urban, and rural areas, and found 24 \% decrease in the urban area.  \cite{de2019impact} and  \cite{kolarova2021impact} found 26 \% and  41 \%  VTTS reduction, respectively. The range of the VOT effect is from 24 \% to 41 \%.

As \cite{van2016autonomous} indicated,  the network capacity effect is still an open question in literature. The literature showed that an increase in capacity ranges from $1$ \% to $414$ \% (see Table 1 in \cite{van2016autonomous}). At network level, \cite{lu2020impact} showed that the network capacity increases by 16 \% and \cite{moshahedi2022macroscopic} found the enhancement of 30 \% in the outflow. However, these studies do not use the real dataset. \cite{huang2023characterizing} found the network capacity effect ranging from 2.9  \%  to 19 \%  through simulation analysis calibrated by  public autonomous vehicle datasets.

\subsection{Commuter preferences}
We next incorporate the commuting cost in the presence of hypercongestion into commuter preferences. The utility of downtown commuters is expressed as the following Cobb--Douglas utility function: 
\begin{align}
 & u_d(z_d(t), a_d(t)) = \{z_d(t)\}^{1-\mu} \{a_d(t)\}^\mu,    
\end{align}
where $\mu \in (0,1)$, $z_d(t)$ represents the consumption of the num\'{e}raire good and $a_d(t)$ represents the lot size of housing which downtown commuters consume.  The budget constraint is 
\begin{align}
& w = z_d(t) + ( r_d + r_A ) a_d(t) + C_d(t),
\end{align}
where $w$ represents their income, $r_A > 0$ represents the exogenous agricultural rent, and $r_d + r_A$ represents the downtown land rent. The first-order conditions of the utility maximization problem are
\begin{subequations}
\begin{align} \label{eq:first_order_downtown}
& z_d(t) = \left( 1 - \mu \right) y_d(t), \\ 
& a_d(t) = \frac{\mu y_d(t)}{r_d + r_A}, \\
& y_d(t) = w - c_d(t),
\end{align}
\end{subequations} 
where $y_d(t)$ represents the income net of commuting cost earned by a downtown commuter who  arrives at work at time $t$. Substituting this into the utility function, we obtain the indirect utility function as
\begin{align}\label{eq:indirect_utility_downtown}
U_d\left(y_d(t), r_d+r_A\right) = (1-\mu)^{1-\mu} \mu^\mu y_d(t) \left( r_d + r_A \right)^{-\mu}.
\end{align}

Similarly, we  formulate a land use model for the suburban commuters. Their utility and budget constraint are respectively expressed as
\begin{align}
 & u_s(z_s(x,t), a_s(x,t)) = \{z_s(x,t)\}^{1-\mu} \{a_s(x,t)\}^\mu, \\
 & w = z_s(x,t) + ( r_s(x) + r_A ) a_s(x,t) + C_s (x,t),
\end{align}
where $z_s(x,t)$  represents the consumption of the num\'{e}raire good, $a_s(x,t)$ represents the lot size of housing consumed by the suburban commuters at $x$, and $r_s(x) + r_A$ represents the land rent at $x$. Then, the first-order conditions of the utility maximization problem are
\begin{subequations}
\begin{align} \label{eq:first_order_suburban}
& z_s(x,t) = \left( 1 - \mu \right) y_s(x,t), \\
& a_s(x,t) = \frac{\mu y_s(x,t)}{r_s(x) + r_A}, \\ 
& y_s(x,t) = w - c_s(x,t), 
\end{align}
\end{subequations}
where $y_s(x,t)$ represents the income net of commuting cost earned by a commuter who resides at $x$ and arrives at work at time $t$. The indirect utility function is 
\begin{align}\label{eq:indirect_utility_suburban_commuters}
& U_s\left(y_s(x,t), r_s(x)+r_A\right) =  (1-\mu)^{1-\mu} \mu^\mu y_s(x,t) \left( r_s(x) + r_A \right)^{-\mu}.
\end{align}

\section{Equilibrium}
\subsection{Equilibrium conditions}
We assume that the commuters decide their residential locations in the long-run, whereas they choose their trip timings in the short-run, which is similar to the modeling of \cite{gubins2014dynamic},  \cite{takayama2017bottleneck}, and \cite{takayama2020gains}. That is, the downtown commuters and suburban commuters at location $x$ minimize commuting cost $C_d(t)$ and $C_s(x,t)$, respectively by selecting their arrival time $t$ at work taking their residential locations as given in the short-run. Each commuter chooses a residential location (the downtown area or location $x$ in the suburban area) so as to maximize his/her utility in the long-run. 

\subsubsection{Short-run equilibrium conditions} \label{section:short-run_equilibrium}
In the short-run, the commuters only decide their trip timing, which implies that the numbers of downtown commuters and suburban commuters residing at $x$ are assumed to be fixed. The short-run equilibrium conditions differ according to their residential locations. According to Eq.~(\ref{eq:commuting_cost_downtown}), all downtown commuters  arrive at $t=t^*$, and they incur a constant commuting cost  $\alpha T_d$. Therefore, the short-run equilibrium cost $C_d^*$  expressed as  
\begin{align}
& C_d^{*}  =  \alpha T_d. \label{eq:short-run_cost_downtown}
\end{align}

As for the suburban commuters,  Eq.~(\ref{eq:commuting_cost_suburban}) states that the commuting cost consists of the cost $\alpha \tau x$ of the suburban free-flow travel time,  which depends only on the residential location $x$, and the bathtub cost $C_s^b(t)$ owing to the travel time in the downtown area and the schedule delay cost, which depends only on the arrival time $t$ (i.e., $C_s^b(t)=\alpha T(t) + s(t)$). Thus, these commuters choose their arrival time such that the bathtub cost $C_s^b(t)$ is minimized. That is, no suburban commuter residing at $x$ can reduce their bathtub cost by changing their departure time at the short-run equilibrium. The equilibrium conditions are 
\begin{subequations}\label{eq:short_run_equilibrium_condtion}
\begin{align} 
&\begin{cases}
C_s^b(t) = C_s^{b*} & \text{if   } n(t) > 0 \\
C_s^b(t) \geq C_s^{b*} & \text{if   } n(t) = 0 \\
\end{cases}  \qquad  \forall t \in \mathbb{R},  \label{eq:UE_con_car}\\ 
&\int_{t\in\mathbb{R} } \frac{n(t)v(t)}{L} dt = N_s,  \label{eq:N_s} 
\end{align}
\end{subequations}
where $C_s^{b*}$ is the short-run equilibrium bathtub cost of the suburban commuters and $N_s$ is the suburban population. 

Condition (\ref{eq:UE_con_car}) states that if the bathtub cost at time $t$ is greater than the equilibrium bathtub cost, then no one will arrive at their destination at time $t$. Condition (\ref{eq:N_s}) is the conservation law for traffic demand in the suburban area. From these conditions, the congestion dynamics and the short-run equilibrium bathtub cost ($n(t)$ and $C_s^{b*}$, respectively) are endogenously determined at the short-run equilibrium. The short-run equilibrium cost $C_s^{*}(x)$ is expressed as 
 expressed as
\begin{align}
C_s^{*}(x) =    C_s^{b*} + \alpha \tau x.
\end{align}

\subsubsection{Long-run equilibrium conditions}
Each commuter chooses a residential location (the downtown area or a location $x$ in the suburban area) to maximize their indirect utility  in the long run. Note that each commuter's residential location is determined based on the short-run equilibrium state (i.e., commuting cost). Therefore, 
as the short-run equilibrium cost $C_d^*$ of the downtown commuters is given by Eq.~(\ref{eq:short-run_cost_downtown}), the income net of downtown commuting cost is 
\begin{align} \label{eq:income_net_downtown}
& y_d  = w - \alpha T_d.
\end{align}
Similarly, as the short-run equilibrium bathtub cost depends on the total number of suburban commuters, the income net of suburban commuting cost at $x$ is expressed as 
\begin{align}
    y_s(x) = w - C_s^{b*} (N_s) - \alpha  \tau x. 
\end{align}

The equilibrium conditions are 
\begin{subequations} \label{eq:long_run_condition}
\begin{align}
&\begin{cases}
U_d(y_d, r_d+r_A)  = U^* & \text{if   } N_d > 0 \\
U_d(y_d, r_d+r_A)  \leq  U^* & \text{if   } N_d = 0 \\
\end{cases},  \label{eq:long1} \\ 
&\begin{cases}
U_s(y_s(x), r_s(x)+r_A)  = U^* \text{ if } N_s(x) > 0 \\
U_s(y_s(x), r_s(x)+r_A)  \leq  U^*  \text{ if } N_s(x) = 0 \\
\end{cases} \label{eq:long2}   \forall x \in \mathbb{R}_+,  \\ 
&\begin{cases}
a_d(y_d, r_d+r_A) N_d  = A_d & \text{if   } r_d > 0 \\
a_d(y_d, r_d+r_A) N_d   \leq A_d & \text{if   } r_d = 0 \\
\end{cases},  \label{eq:long3} \\ 
&\begin{cases}
a_s(y_s(x), r_s(x)+r_A) N_s(x)  = A_s(x) & \text{if   } r_s(x) > 0 \\
a_s(y_s(x), r_s(x)+r_A) N_s(x)   \leq A_s(x) & \text{if   } r_s(x) = 0 \\
\end{cases}\label{eq:long4}   \forall x \in \mathbb{R}_+,  \\ 
& N_d + \int_{x \in  \mathbb{R}_+} N_s(x) dx = N,    \label{eq:N} 
\end{align}
\end{subequations}
where $U^*$ is the long-run equilibrium utility level and $a_d(y_d, r_d+r_A)$ and $a_s(y_s(x), r_s(x)+r_A)$ denote the lot sizes of downtown commuters and suburban commuters who reside at location $x$, respectively.  These lot sizes are given by 
\begin{subequations}
\begin{align} \label{eq:lot_size_equilibrium}
&a_d(y_d, r_d+r_A) =  \frac{\mu y_d}{r_d + r_A}, \\
&a_s(y_s(x), r_s(x)+r_A) = \frac{\mu y_s(x)}{r_s(x) + r_A}.
\end{align}
\end{subequations}

Conditions (\ref{eq:long1}) and  (\ref{eq:long2}) are the equilibrium conditions for the downtown and suburban residential location choices. respectively. These conditions state that no commuter has the incentive to change residential locations unilaterally. Conditions (\ref{eq:long3}) and (\ref{eq:long4}) are the land market clearing conditions, which states that if the total land demand for housing equals the land size, then land rent is (weakly) larger than the agricultural rent $r_A$. Condition (\ref{eq:N}) shows the population constraint. 

\subsection{Equilibrium properties}
\subsubsection{Short-run equilibrium properties}
As the short-run equilibrium condition (\ref{eq:short_run_equilibrium_condtion}) coincides with those in the bathtub model of \cite{small2003hypercongestion}, we have the following properties. 
\begin{lemma} \label{lemma:short_run_without}
The short-run equilibrium has the following properties 
\begin{itemize}
    \item The short-run equilibrium bathtub cost satisfies 
\begin{align} \label{eq:equilibrium_bathtub_cost}
F(\theta) \equiv N_s -  \alpha n_j \left( \frac{1}{\beta} + \frac{1}{\gamma} \right) \left( \ln \theta + \frac{1}{\theta} - 1 \right) = 0,
\end{align} 
where  $\theta \equiv \frac{C_s^{b*} v_f}{\alpha L }$.  
\item  Hypercongestion exists if $\theta>2$. 
\end{itemize}
\end{lemma}
\begin{proof}
    See Appendix A. 
\end{proof}    

As cases without hypercongestion are out of our interest, we impose the assumption that $\theta > 2$ from Lemma~\ref{lemma:short_run_without}. All variables except the equilibrium bathtub cost $C_s^{b*}$ are exogenous because the total number of suburban commuters $N_s$ is given in the short-run. As the function of $F(\theta)$ is strictly monotone with respect to $\theta$ when $\theta > 2$, the equilibrium bathtub cost $C_s^{b*}$ is uniquely determined.

The replacement of all human driven vehicles with autonomous vehicles influences the equilibrium bathtub cost.  From Eq.~(\ref{eq:equilibrium_bathtub_cost}), we obtain the following equation:
\begin{align} \label{eq:comparative_1}
\frac{ {\rm d} C_s^{b*}}{{\rm d} n_j} = - \left(\frac{\partial \theta}{\partial C_s^{b*} } \right)^{-1} \left( \frac{\partial F}{\partial\theta} \right)^{-1} \left( \frac{\partial F}{\partial n_j} \right) < 0.
\end{align}
As $\left(\frac{\partial \theta}{\partial C_s^{b*} } \right)^{-1}$ is positive and both of $ \left( \frac{\partial F}{\partial\theta} \right)^{-1}$ and $\left( \frac{\partial F}{\partial n_j} \right)$ are negative, the network capacity effect (i.e., the increase in $n_j$) always decreases the equilibrium bathtub cost. This is because the network can process more vehicles due to the capacity expansion, so more commuters arrive at their destinations near their desired arrival times.

We also obtain the following equation from Eq.~(\ref{eq:equilibrium_bathtub_cost}), which indicates 
that {\it the VOT effect (i.e., the reduction in $\alpha$) may increase or decrease the equilibrium bathtub cost:} 
\begin{align} \label{eq:comparative_2_takayama}
    &\frac{{\rm d} C_s^{b*}}{{\rm d}\alpha} = - \left(\frac{\partial \theta}{\partial C_s^{b*} } \right)^{-1} \left( \frac{\partial F}{\partial\theta} \right)^{-1} \left( \frac{\partial F}{\partial \alpha} \right) \gtrless 0 \quad \text{if} \quad  \frac{\partial F}{\partial \alpha}  \lessgtr 0,
    \\
    &
  \frac{\partial F}{\partial \alpha} = \frac{\partial f(\theta, \alpha)}{\partial \alpha} + \frac{\partial f(\theta, \alpha)}{\partial \theta}\frac{\partial \theta}{\partial \alpha},
  \label{eq:df_da_takayama}
\end{align}
where $f(\theta, \alpha)=-  \alpha n_j \left( \beta^{-1} + \gamma^{-1} \right) \left( \ln \theta + \theta^{-1} - 1 \right)$
is the second term of $F(\theta)$. 
The first term on the RHS of Eq.(\ref{eq:df_da_takayama}) is negative and means that the VOT reduction simply decreases the bathtub congestion cost. The second term on the RHS of Eq.(\ref{eq:df_da_takayama}) is positive and means that a higher temporal concentration of traffic demand near the desired arrival time causes a more severe capacity drop and results in an increase in the bathtub congestion cost. Therefore, whether the equilibrium bathtub cost increases or decreases depends on the magnitudes of the network capacity and VOT effects.

Although the network capacity effect on hypercongestion  has been extensively studies \cite[e.g.,][]{lu2020impact, moshahedi2022macroscopic}, the VOT effect on hypercongestion has been insufficiently explored. Most of the existing works investigating the network capacity effect consider autonomous vehicles to be beneficial \cite[e.g.,][]{fagnant2015preparing, van2016autonomous, lamotte2017use}. The only exception is  \cite{van2016autonomous}, who  examined both the effects of network capacity and VOT in the standard bottleneck model and showed that they {\it always} decrease the equilibrium travel cost. However, this may not hold in the presence of hypercongestion due to the capacity drop caused by the VOT effect. In other words, hypercongestion can worsen and the travel cost can increase when autonomous vehicles are introduced if the VOT effect is large.

The short-run equilibrium properties are summarized as follows. 

\begin{proposition} \label{prop:short_run_property}
The short-run equilibrium has the following properties. 
\begin{itemize}
    \item The short-run equilibrium bathtub cost  is uniquely determined. 
    \item The introduction of autonomous vehicles may increase or decrease the short-run equilibrium bathtub cost in the presence of hypercongestion. 
\end{itemize}
\end{proposition}

\subsection{Long-run equilibrium properties}
We examine the properties of the urban spatial structure at the long-run equilibrium. 
\begin{lemma} \label{lemma:long_run_equilibrium_property0}
At the long-run equilibrium, 
\begin{align}\label{eq:indirect_utility_suburban}
&\frac{d \{r_s(x)+r_A\}}{dx} = - \alpha \tau \frac{ r_s(x)+r_A}{\mu y_s(x)} < 0, \\
&\frac{d \left\{ \frac{N_s(x)}{A_s(x)} \right\} }{dx} = - \alpha \tau \frac{(1-\mu) (r_s(x)+r_A)}{ \{ \mu \left(w - C_s^{b*} - \alpha \tau x\right)\}^2} < 0.
\end{align}
\end{lemma}
\begin{proof}
    See Appendix B. 
\end{proof}

Lemma \ref{lemma:long_run_equilibrium_property0} indicates that the land rent and population density in the suburban area decrease with an increase in location $x$. These properties are consistent with those in the literature \cite[e.g.,][]{fujita1989urban}. Then, the long-run equilibrium properties are as follows. 

\begin{lemma} \label{lemma:long_run_equilibrium_property}
At the long-run equilibrium,
\begin{itemize}
    \item $N_d^*$ and $N_s^*(x)$ are given by 
     \begin{align} 
            &N_d^* =\frac{1}{\mu} \{y_d\}^{\frac{1-\mu}{\mu}} \left( w - C_s^{b*} \right)^{-\frac{1}{\mu}} \left( r_s(0) + r_A \right)  A_d, \label{eq:long_run_equilibrium_downtown_commuters}\\
            &N_s^*(x) = \frac{1}{\mu} \{w - C_s^{b*} - \alpha \tau x \}^{\frac{1-\mu}{\mu}} \{  y_d \}^{-\frac{1}{\mu}} \left( r_s(0) + r_A \right)  A_s(x).\label{eq:long_run_equilibrium_suburban_commuters}
        \end{align}
    \item The city boundary $x_f$ is given by 
     \begin{align}\label{eq:long_run_equilibrium_city_boundary}
            x_f = \frac{w-C_s^{b*}}{\alpha \tau} \left( \{ r_A \}^{-\mu} - \{ r_s(0) + r_A \}^{-\mu} \right)\{ r_A \}^{\mu}.
        \end{align}
    \item  $r_s(0)$ is determined from $\int_0^{x_f} N_s^*(x)dx=N - N_d^*$. 
\end{itemize}   
\end{lemma}
\begin{proof}
    See Appendix C. 
\end{proof}

Eq.~(\ref{eq:long_run_equilibrium_downtown_commuters}) states that the downtown population decreases with a decrease in the short-run equilibrium bathtub cost.  This occurs because a lower commuting cost encourages more commuters to choose the suburban area as their residential location. 
Combining Lemma~\ref{lemma:long_run_equilibrium_property} with Proposition~\ref{prop:short_run_property} suggests that the introduction of autonomous vehicles does not necessarily result in an increase in the suburban population (i.e., suburbanization). That is, {\it the downtown population may increase due to the VOT effect, and thus the total car traffic demand decreases in the long-run.}

It should be noted that the VOT effect of autonomous vehicles decreases the cost of free-flow travel time in the suburban area, causing suburban commuters to reside farther from the downtown. This leads to a decrease in the population density in the suburban area and the spatial expansion of the city, as indicated by Eqs.~(\ref{eq:long_run_equilibrium_suburban_commuters}) and (\ref{eq:long_run_equilibrium_city_boundary}). This is consistent with the standard results obtained in the literature. Therefore, {\it  even if the introduction of autonomous vehicles leads to a decrease in the suburban population, the city can spatially expand outward. } In Section 6, we will demonstrate that such a case actually occurs. 

The results are summarized as follows. 
\begin{proposition}
\ 
\begin{itemize}
    \item The introduction of autonomous vehicles (i.e., the VOT and network capacity effects) may increase the short-run equilibrium bathtub cost, resulting in a decrease in the suburban population (i.e., total car traffic demand) in the long-run.  
    \item Even if the introduction of autonomous vehicles decreases the suburban population,  the city may spatially expand outward.
\end{itemize}    
\end{proposition}

\section{Perimeter Control}
\subsection{Perimeter control scheme}
During perimeter control, the inflow rate to the downtown area is restricted to maintain the maximum throughput in the downtown area. A queue will develop outside the perimeter boundary if the arrival rate at the boundary exceeds the restricted inflow rate to the downtown area. 

We  derive the critical accumulation $n_{cr}$, where the throughput is maximized from Eq.~(\ref{eq:speed}). 
\begin{align} \label{eq:critical_accumulation}
    n_{cr} = \frac{n_j}{2}.
\end{align}
The aim of perimeter control is to restrict the inflow at the perimeter boundary so as not to exceed the critical accumulation in the downtown area. Therefore, the inflow rate during perimeter control $I (t)$ is determined by 
\begin{align}
    I(t) = 
    \begin{cases}
        I_p & \qquad  \text{if } n(t) = n_{cr}  \\
        A_b(t) & \qquad  \text{if } n(t) < n_{cr}
    \end{cases},
\end{align}
where $I_p$ is the inflow rate during perimeter control and $A_b(t)$ is the arrival rate at the perimeter boundary at time $t$. If the accumulation is below the critical level, then  no restriction is implemented; all vehicles at the boundary can enter the downtown area. Once accumulation reaches the critical accumulation, the inflow rate is restricted to $I_p$.  The throughput can be maximized during perimeter control if the inflow rate $I(t)$ is set to the NEF at the critical accumulation. Thus, from Eqs.~(\ref{eq:speed})-(\ref{eq:NEF}) and (\ref{eq:critical_accumulation}), we have 
\begin{align}
    I_p = \frac{n_jv_f}{4L}.
\end{align}

\subsection{Queuing dynamics at perimeter boundaries}
A queue will develop outside the perimeter boundary if the arrival rate at the perimeter boundary exceeds the inflow rate $I_p$. We model the queuing dynamics as a point queue and assume a first-arrived-first-in property. Therefore, the waiting time of a commuter who arrives at their destination at time $t$, $T_w(t)$, is 
\begin{align} \label{eq:waiting_time_boundary}
    T_w(t) = \frac{q(t)}{I_p},
\end{align}
where $q(t)$ is the number of vehicles queued at the perimeter boundary when a commuter who arrives at their destination at time $t$ reaches the boundary.

\subsection{Schedule preferences of suburban commuters}
The schedule preferences of the downtown commuters are the same as those without perimeter control in Eq.~(\ref{eq:travel_cost}). Given the queue dynamics during perimeter control, suburban commuting cost $C_s^p(x,t)$ of a commuter who resides at $x$ and arrives at work at time $t$ is
\begin{subequations}\label{eq:gc_perimeter}
\begin{align}
&C_s^p(x, t) =  \alpha \left( T^p(t) + \tau x \right)  + s(t) ,\label{eq:commuting_cost_perimeter}  \\
& T^p(t)  =
\begin{cases}
T(t)   & \qquad  \text{if } t \leq t_s^p \\
\frac{L}{v_f/2} + T_w(t) & \qquad  \text{if } t_s^p < t \leq t_e^p \\
T(t)   & \qquad  \text{if }  t_e^p < t \\
\end{cases}, \\
& s(t) = 
\begin{cases}
\beta \left( t^* - t \right) & \qquad  \text{if } t \leq t^*  \\
\gamma \left( t -  t^* \right) & \qquad \text{if }t > t^* 
\end{cases}, \label{eq:schedule_delay_perimeter}
\end{align}
\end{subequations}
where $T^p(t)$ is the downtown travel time under perimeter control, which includes the travel time in the downtown area and the waiting time at the perimeter boundary.
$t_s^p$ and $t_e^p$ are the start and end times of perimeter control, respectively. The difference between the schedule preferences at the equilibrium without and with perimeter control appears in the downtown travel time under perimeter control $T^p(t)$. The downtown travel time before and after perimeter control implementation ($t \leq t_s^p$ and $ t_e^p < t $, respectively) are the same as Eq.~(\ref{eq:travel_cost}). During perimeter control, the travel time in the downtown area is $\frac{L}{v_f/2}$ because the NEF is maintained at the maximum. Furthermore, waiting time at the perimeter boundary (Eq.~(\ref{eq:waiting_time_boundary})) is  incurred in addition to the travel time in the downtown area. 

\section{Equilibrium under Perimeter Control}
\subsection{Equilibrium conditions}
\subsubsection{Short-run equilibrium}
The short-run equilibrium conditions for the downtown commuters are the same as those without perimeter control (Section \ref{section:short-run_equilibrium}). Therefore,
\begin{align}
& C_d^{p*}  =  \alpha T_d, 
\end{align}
where $C_d^{p*} $ is  the short-run equilibrium cost of the downtown commuters. 

As at the short-run equilibrium without perimeter control, the suburban commuters choose their arrival times such that  the bathtub cost ($C_s^{bp}(t)= \alpha T^p(t) + s(t)$) is minimized. Therefore, the equilibrium conditions are 
\begin{subequations}\label{eq:short_run_equilibrium_condtion_perimeter}
\begin{align} 
&\begin{cases}
C_s^{bp}(t) = C_s^{bp*} & \text{if   } n(t) > 0 \\
C_s^{bp}(t) \geq C_s^{bp*} & \text{if   } n(t) = 0 \\
\end{cases}  \qquad  \forall t \in \mathbb{R},  \label{eq:UE_con_car_perimeter}\\ 
&\begin{cases}
n(t) = \frac{n_j}{2} & \text{if   } q(t) > 0 \\
n(t) \leq \frac{n_j}{2}  & \text{if   } q(t) = 0 \\
\end{cases}  \qquad  \forall t \in \mathbb{R},  \label{eq:boundary_condition_perimeter}\\ 
& \int_{t\in\mathbb{R} } \frac{n(t)v(t)}{L} dt = N_s,  \label{eq:N_s_perimeter} 
\end{align}
\end{subequations}
where $ C_s^{bp*}$ is the short-run equilibrium bathtub cost during perimeter control and $N_s$ is the number of suburban commuters under perimeter control.

Conditions (\ref{eq:UE_con_car_perimeter}) and (\ref{eq:N_s_perimeter}) are the same as Conditions  (\ref{eq:UE_con_car}) and (\ref{eq:N_s}), respectively. Condition (\ref{eq:boundary_condition_perimeter}) reflects the restriction of the inflow to the downtown area during perimeter control; accumulation is at the critical level if there is a queue at the perimeter boundary. Otherwise, accumulation is lower than the critical level. Then, the short-run equilibrium cost $C_s^{p*}(x)$ is 
\begin{align}
C_s^{p*}(x) =    C_s^{bp*} + \alpha \tau x. 
\end{align}

\subsubsection{Long-run equilibrium}
In the long-run, the difference between the cases with and without perimeter control appears only in the income net of suburban commuting cost. 
Specifically, under the perimeter control, the income net of suburban commuting cost  at location $x$ is 
\begin{align}
     y^p_s(x) = w - C_s^{p*}(x). 
\label{eq:income_net_suburb_pc}
\end{align}
The long-run equilibrium conditions are thus represented as (\ref{eq:long_run_condition}) with the use of (\ref {eq:income_net_suburb_pc}).

\subsection{Equilibrium properties}
\subsubsection{Short-run equilibrium}
Similar to the model of \cite{dantsuji2022perimeter}, we have the following properties. 
\begin{lemma} \label{lemma:short-run_perimeter}
    The short-run equilibrium under perimeter control has the following properties
    \begin{itemize}
        \item The short-run equilibrium bathtub cost satisfies 
        \begin{align} \label{eq:equilibrium_bathtub_cost_perimeter}
        F^p(\theta^p)  \equiv N_s - \alpha n_j \left( \frac{1}{\beta} + \frac{1}{\gamma} \right) \left( \frac{\theta^p}{4} + \ln 2 - 1 \right)  = 0,
        \end{align}
        where $\theta^p \equiv \frac{C_s^{pb*} v_f}{\alpha L }$.  
        \item A queue develops at the perimeter boundary, and its length increases toward the desired arrival time. 
    \end{itemize}
\end{lemma}
\begin{proof}
    See Appendix D. 
\end{proof}    

All variables except the equilibrium bathtub cost $C_s^{pb*} $ are exogenous. The equilibrium bathtub cost $C_s^{pb*}$ is uniquely determined because the function of $F^p(\theta^p)$ is strictly monotone with respect to $\theta^p$.   Eq.~(\ref{eq:equilibrium_bathtub_cost_perimeter}) is rewritten as
\begin{align} \label{eq:equilibrium_bathtub_cost_perimeter_rev}
    C_s^{pb*} =  \frac{\beta \gamma}{\beta + \gamma} \frac{N_s}{ \frac{n_jv_f}{4L} } + \frac{4 \alpha L}{v_f} \left( 1 - \ln 2 \right). 
\end{align}

When autonomous vehicles are introduced, as shown in Eq.~(\ref{eq:equilibrium_bathtub_cost_perimeter_rev}), {\it both the network capacity and VOT effects decrease the  short-run equilibrium cost of the suburban commuters under perimeter control} (i.e., ${\rm d}  C_s^{pb*} / {\rm d}  n_j <0 $ and ${\rm d}  C_s^{pb*} / {\rm d}  \alpha >0 $, respectively). The network capacity effect increases the number of vehicles traveling in the downtown area during perimeter control and the restricted inflow rate at the perimeter boundary. Consequently,  more suburban commuters arrive at their destinations near their desired arrival times. This decreases the short-run equilibrium cost of the suburban commuters.

Without perimeter control, the VOT effect may increase the short-run equilibrium cost for the suburban commuters  due to  capacity drop. With perimeter control, the VOT effect always decreases the short-run equilibrium cost because capacity drop never occurs. Even though the queue length at the perimeter boundary increases due to the VOT effect, the equilibrium cost decreases because the inflow rate is constant regardless of the queue length. This mechanism  is  the same as that in the standard bottleneck model  \citep{van2016autonomous} because this situation involves  a bottleneck with a fixed capacity (i.e., the value of NEF at critical accumulation) between the downtown and suburban areas. Note that the first term on the  RHS of Eq.~(\ref{eq:equilibrium_bathtub_cost_perimeter_rev}) represents the bottleneck cost where the bottleneck capacity is $\frac{n_jv_f}{4L}$, which is identical to the value of NEF at critical accumulation. This indicates that the impacts of autonomous vehicles in the presence of hypercongestion contradict those in the standard bottleneck model. 

Next, we compare equilibria with and without perimeter control to demonstrate the effects of perimeter control on the urban spatial structure. The only difference between the cases with and without perimeter control appears in the short-run equilibrium bathtub cost ($C_s^{b*}$ and $C_s^{bp*}$, respectively). Eqs.~(\ref{eq:equilibrium_bathtub_cost}) and (\ref{eq:equilibrium_bathtub_cost_perimeter}) show that the short-run equilibrium bathtub cost is reduced by perimeter control when hypercongestion exists (see Appendix E for the proof). The short-run equilibrium bathtub cost decreases because network capacity drop  never occurs under perimeter control. Thus, more commuters  arrive at their destinations near their desired arrival times than at user equilibrium where  capacity drop occurs. Therefore, {\it although queuing congestion occurs at the perimeter boundary, the short-run equilibrium bathtub cost decreases. }

The Properties of the short-run equilibrium with perimeter control are summarized as follows. 
\begin{proposition} \label{prop:short_run_property_under_perimeter_control}
The short-run equilibrium with perimeter control has the following properties.
    \begin{itemize}
    \item The short-run equilibrium bathtub cost with perimeter control is uniquely determined. 
    \item The introduction of autonomous vehicles decreases the short-run equilibrium bathtub cost under perimeter control. 
    \item Hypercongestion mitigation by perimeter control decreases the short-run equilibrium bathtub cost.
    \end{itemize}
\end{proposition}

\subsubsection{Long-run equilibrium}
We investigate the properties of the urban spatial structure at the long-run equilibrium with perimeter control. In the long-run, the difference between cases with and without perimeter control is the income net of commuting cost. Thus, we have the following properties. 
\begin{lemma} \label{lemma:long_run_equilibrium_property_perimeter}
At the long-run equilibrium under perimeter control
    \begin{itemize}
    \item $N_d^{p*}$ and $N_s^{p*}(x)$ are given by 
     \begin{align} \label{eq:long_run_equilibrium_downtown_commuters_perimeter}
            &N_d^{p*} = \frac{1}{\mu} \{y_d\}^{\frac{1-\mu}{\mu}} \left( w - C_s^{bp*} \right)^{-\frac{1}{\mu}} \left( r_s(0) + r_A \right)  A_d, \\
            &N_s^{p*}(x) = \frac{1}{\mu} \{w - C_s^{bp*} - \alpha \tau x \}^{\frac{1-\mu}{\mu}} \{  y_d \}^{-\frac{1}{\mu}} \left( r_s(0) + r_A \right)  A_s(x). \label{eq:long_run_equilibrium_suburban_commuters_perimeter}
        \end{align}
    \item The city boundary $x_f$ is given by 
     \begin{align}\label{eq:long_run_equilibrium_city_boundary_perimeter}
            x_f = \frac{w-C_s^{bp*}}{\alpha \tau} \left( \{ r_A \}^{-\mu} - \{ r_s(0) + r_A \}^{-\mu} \right)\{ r_A \}^{\mu}.
        \end{align}
    \item  $r_s(0)$ is determined from $\int_0^{x_f} N_s^{p*}(x)dx=N - N_d^{p*}$. 
\end{itemize}   
\end{lemma}
 Lemma 5 suggests that the urban spatial structure at the long-run equilibrium under perimeter control has the same properties as the cases without perimeter control. Specifically, Eq. (\ref{eq:long_run_equilibrium_downtown_commuters_perimeter}) indicates that the downtown population decreases with a decrease in the short-run equilibrium bathtub cost. Eqs. (\ref{eq:long_run_equilibrium_suburban_commuters_perimeter}) and (\ref{eq:long_run_equilibrium_city_boundary_perimeter}) state that the VOT effect decreases the population density in the suburban area and expands the city outward.  Therefore, as the introduction of autonomous vehicles always decreases the short-run equilibrium bathtub  cost under perimeter control (Proposition \ref{prop:short_run_property_under_perimeter_control}), {\it the downtown population decreases under perimeter control due to the introduction of autonomous vehicles.}

As the short-run equilibrium bathtub cost is reduced by perimeter control (Proposition \ref{prop:short_run_property_under_perimeter_control}), the downtown population decreases in the long-run, which can be seen from Eq.~(\ref{eq:long_run_equilibrium_downtown_commuters_perimeter}). Since the income net of the suburban commuting cost increases, more commuters select the suburban area as their residential locations. Cities are expanded by hypercongestion mitigation.  

The results are summarized as follows.
\begin{proposition}
\
\begin{itemize}
    \item The introduction of autonomous vehicles under perimeter control decreases the short-run equilibrium bathtub cost,  and  results in a decrease in the downtown population in the long-run.     
    \item Hypercongestion mitigation by perimeter control decreases the short-run equilibrium cost, and results in a decrease in the downtown population in the long-run.
\end{itemize}
\end{proposition}

\section{Numerical Examples}

\begin{figure}[p]
\begin{minipage} {0.5\columnwidth}
  \centering
  \includegraphics[width=\textwidth]{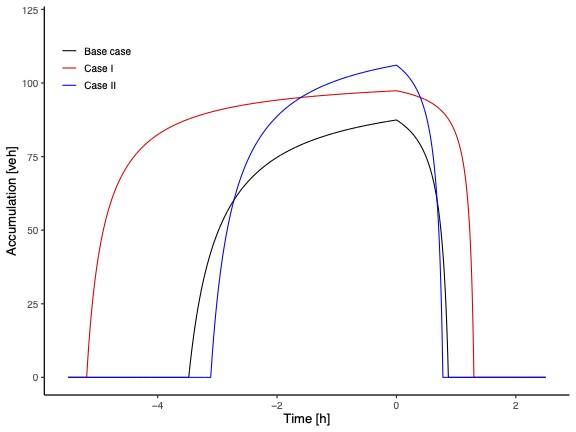}
  \caption{Accumulation in the downtown area}\label{fig:accumulation_short}
\end{minipage}
\begin{minipage} {0.5\columnwidth}
  \centering
  \includegraphics[width=\textwidth]{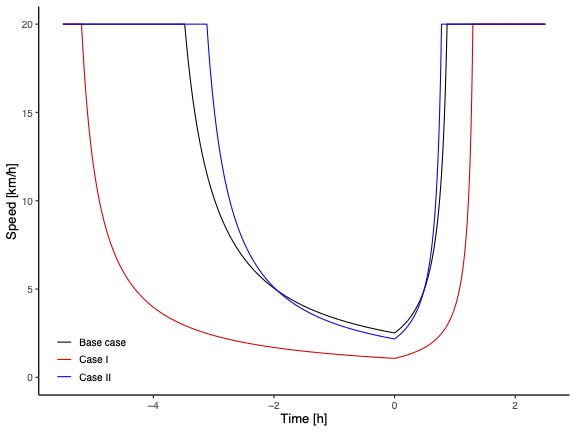}
  \caption{Speed in the downtown area}\label{fig:speed_short}
\end{minipage}
\begin{minipage} {0.5\columnwidth}
  \centering
  \includegraphics[width=\textwidth]{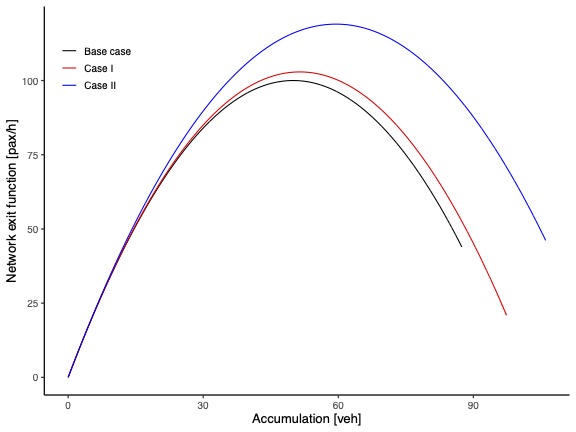}
  \caption{MFD in the downtown area.  }\label{fig:MFD_short}
\end{minipage}
\begin{minipage} {0.5\columnwidth}
  \centering
  \includegraphics[width=\textwidth]{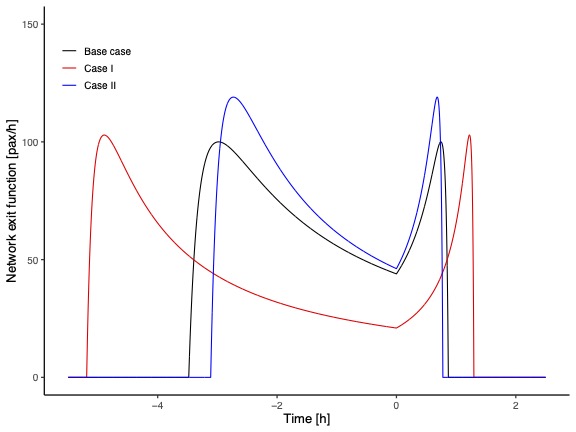}
  \caption{NEF of suburban commuters in downtown area}\label{fig:NEF_short}
\end{minipage}
\begin{minipage} {0.5\columnwidth}
  \centering
  \includegraphics[width=\textwidth]{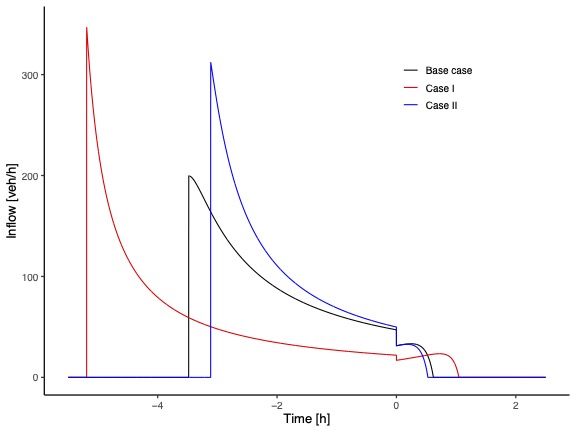}
  \caption{Inflow in downtown area}\label{fig:inflow_short}
\end{minipage}
\end{figure}

\begin{figure}
\begin{minipage}{0.5\columnwidth}
\centering
 \includegraphics[width=\textwidth]{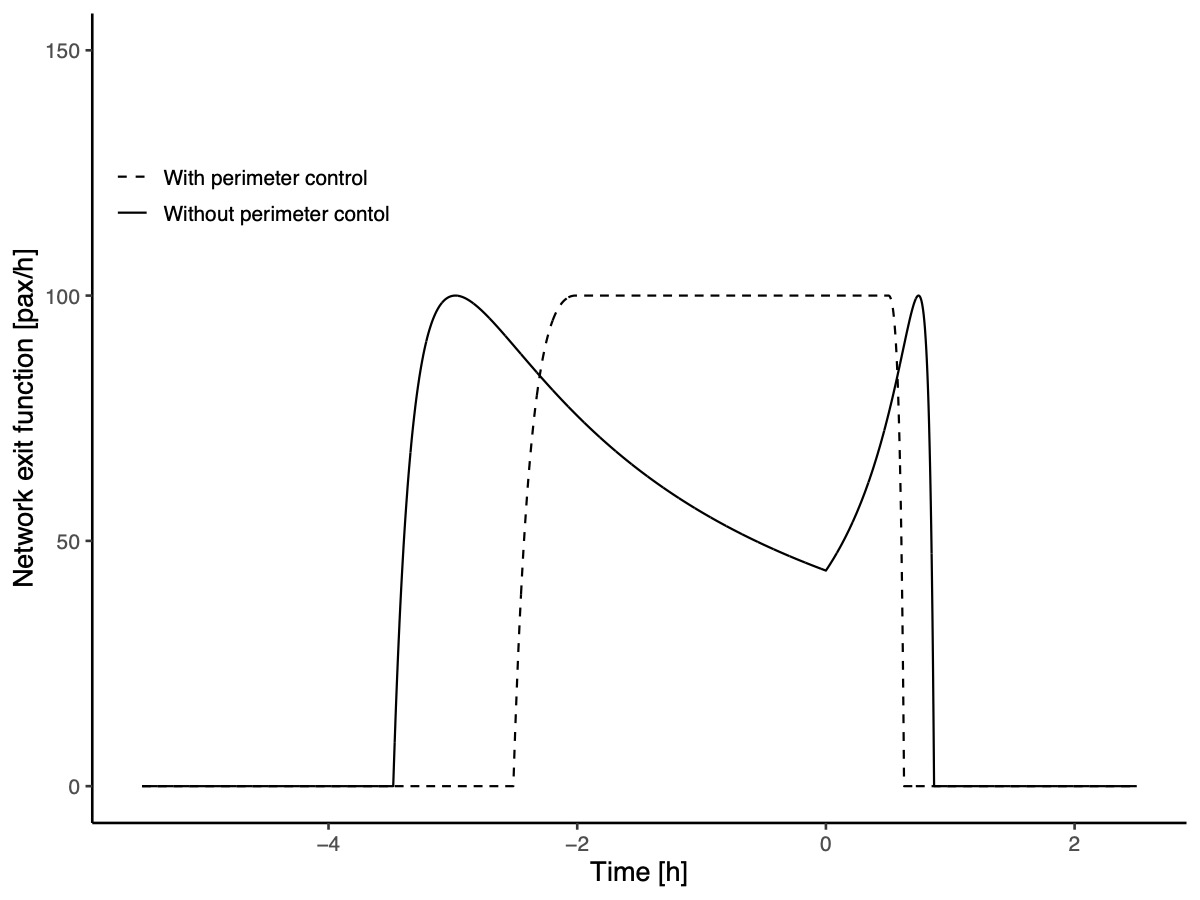}
  \caption{NEF with and without perimeter control (base case)}\label{fig:NEF_peri_base_short}
\end{minipage}
\begin{minipage} {0.5\columnwidth}
  \centering
  \includegraphics[width=\textwidth]{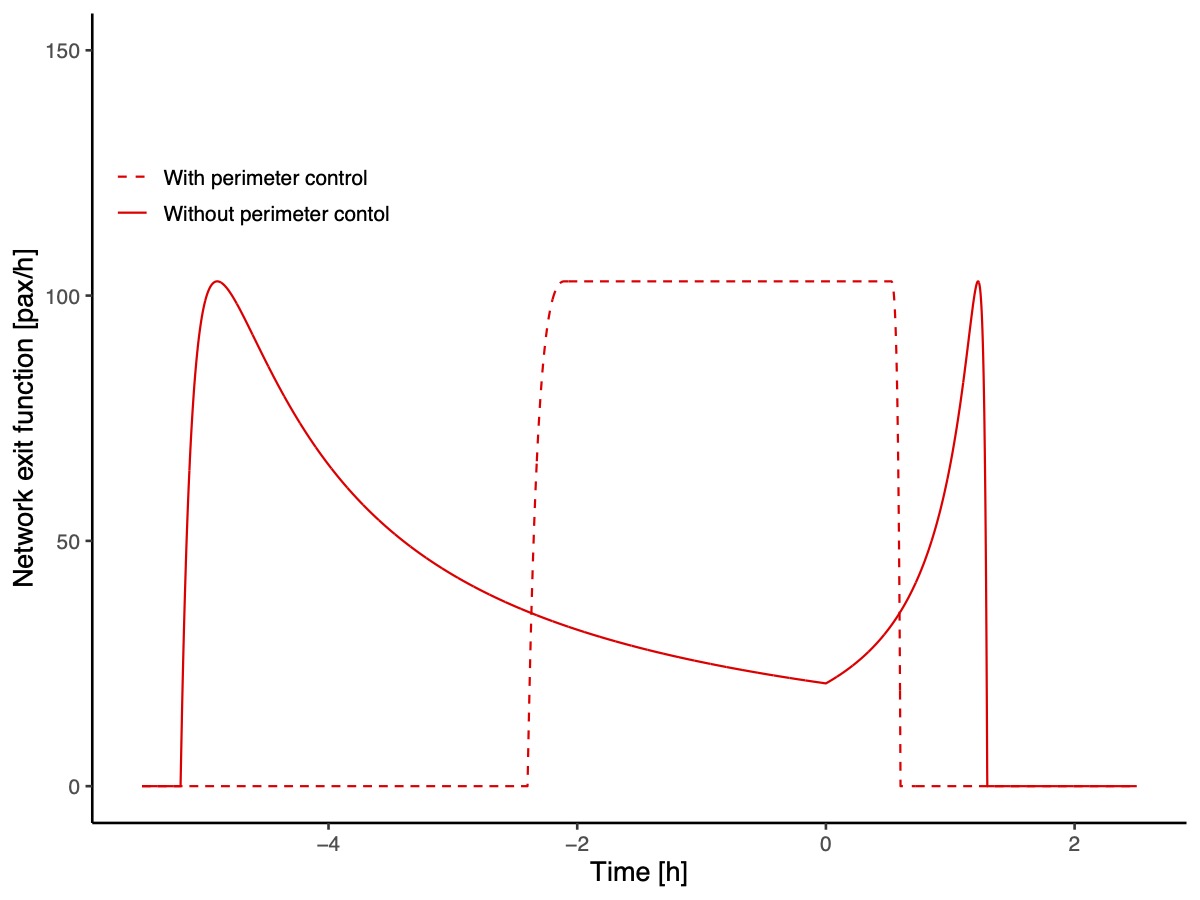}
  \caption{NEF with and without perimeter control (Case I)}\label{fig:NEF_peri_case1_short}
\end{minipage}
\begin{minipage} {0.5\columnwidth}
  \centering
  \includegraphics[width=\textwidth]{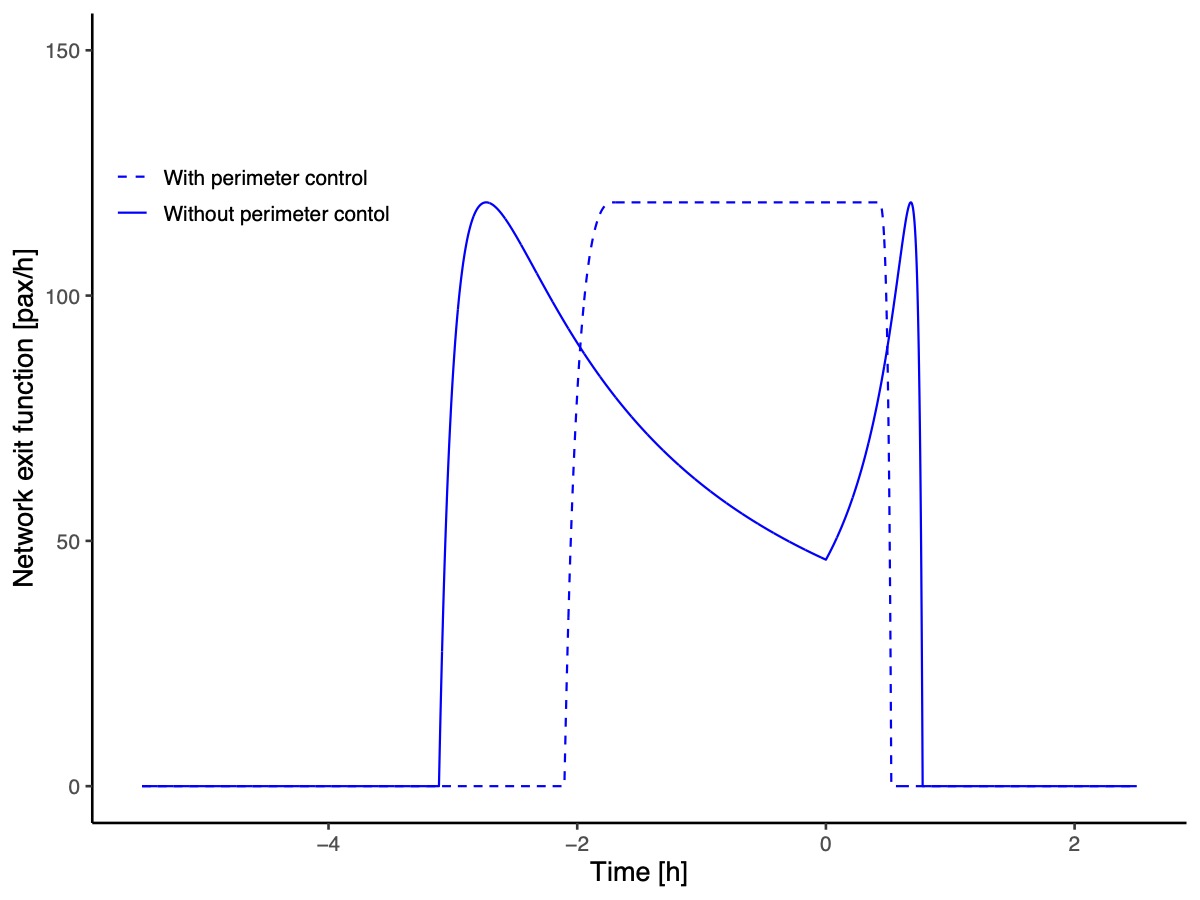}
  \caption{NEF with and without perimeter control (Case II)  }\label{fig:NEF_peri_case2_short}
\end{minipage}
\end{figure}

\begin{table}[t]
  \caption{Numerical results of short-run equilibrium.  }
  \label{table:numerical_result_short}
  \centering
  \begin{tabular}{lccc}
    \hline
    Case  & $C_s^{b*}$ & $C_s^{bp*}$ & $\frac{C_s^{bp*}}{C_s^{b*}}$ \\
    \hline \hline
    Base case ($\eta=1$, $\xi=1$)   & 39.8 & 30.1 & 0.76 \\
    Case I ($\eta=0.59$, $\xi=1.029$)   & 54.8 & 26.9 & 0.49 \\
    Case II ($\eta=0.76$, $\xi=1.19$)   & 34.9 & 24.8 & 0.71 \\
    \hline 
  \end{tabular} 
\end{table}

\subsection{Short-run equilibrium} \label{section:short-run}
First, we numerically investigate the short-run equilibrium patterns to understand the short-run impacts of autonomous vehicles and perimeter control. We set the fixed suburban commuters to $N_s=300$ [pax]. The other parameters related the short-run equilibrium are as follows: $v_f=20$ [mph], $n_j=100$ [veh], $\alpha, \beta, \gamma = 20, 10, 40$ [\$/h], $t^*=0$, $L=5$ [mile]. Based on the empirical studies on the impacts of autonomous vehicles (see the details in Section 2.1), we compare the equilibrium patterns of three cases: the case of base equilibrium (Base case); a case with the highest VOT and lowest network capacity effects of autonomous vehicles (Case I); and a case with the lowest VOT and highest network capacity effects of autonomous vehicles (Case II). In addition to the above parameters, we set the VOT effect parameter to $\eta=0.59$ (Case I) and $0.76$ (Case II), and the network capacity effect parameter to $\xi = 1.029$ (Case I) and $1.19$ (Case II). 

Figs. \ref{fig:accumulation_short} - \ref{fig:inflow_short} show the equilibrium patterns for three cases. In all cases, a similar evolution of car accumulation is obtained as depicted in Fig. \ref{fig:accumulation_short}. The introduction of autonomous vehicles results in more vehicles circulating in the downtown area at the desired arrival time, due to the network capacity effect. The rush hour is longer in Case I and shorter in Case II than that in the Base case. This indicates that the introduction of autonomous vehicles increases and decreases the short-run equilibrium bathtub cost in Case I and II, respectively, as shown in Table \ref{table:numerical_result_short}.

Fig. \ref{fig:speed_short} reveals that the speed at the desired arrival time in Case II is higher than that in Case I, although the maximum accumulation observed at the desired arrival time in Case II is higher than that in Case I due to the higher network capacity effect in Case II. This leads to the higher NEF in Case II when the maximum accumulation is reached, as depicted in Fig. \ref{fig:MFD_short}. Moreover, even though the maximum accumulation in Case II is higher than that in the Base case in the hypercongested regime,  the NEF in Case II is higher than that in the Base case (with human-driven vehicles) at their maximum accumulation.  Therefore, a lighter capacity drop is observed in Case II, as depicted in Fig. \ref{fig:NEF_short}, leading to a shorter rush hour and resulting in a decrease in the short-run equilibrium bathtub cost.

Conversely, the higher temporal concentration of traffic demand (Fig. \ref{fig:inflow_short}), due to the high VOT effect, induces a more severe capacity drop in Case I compared to the Base case and Case II. This results in an increase in the short-run equilibrium bathtub cost. Therefore, the introduction of autonomous vehicles may increase or decrease the short-run equilibrium bathtub cost in the presence of hypercongestion.

Perimeter control implementation shortens the rush hour and results in a decrease in the short-run equilibrium bathtub cost in all cases, as illustrated in Figs. \ref{fig:NEF_peri_base_short} - \ref{fig:NEF_peri_case2_short}. Since hypercongestion never occurs under perimeter control, more commuters arrive at their destinations near their desired arrival times, which results in a shorter rush hour.  Note that the applied perimeter control strategy is similar to that proposed in \cite{daganzo2007urban}. Even if other advanced control strategies, such as reinforcement learning-based approach \citep{chen2022data}, aiming at maximizing throughput in the downtown area  are employed,  the short-run bathtub equilibrium cost would remain the same. This results in the same impact on the long-run equilibrium as the perimeter control strategy applied in this paper. It is also important to note that the control inflow could be influenced by various factors such as the network infrastructure \citep[e.g.,][]{aboudolas2013perimeter}, it can be challenging to set the targeted accumulation precisely at the critical value. Such biases are investigated in Appendix F.

To further investigate the effect of perimeter control, we estimate the ratio of the short-run bathtub equilibrium cost with perimeter control to that without perimeter control as shown in the last column of Table \ref{table:numerical_result_short}. Given that the duration of hypercongestion is  longest and a severe capacity drop occurs in Case I among the three cases, the effect of perimeter control is the largest in Case I. Although the duration of hypercongestion is shorter in Case II compared to the Base case, the effect of perimeter control is larger in Case II due to more temporal concentration of traffic demand (Fig. \ref{fig:inflow_short}) and higher network capacity.

\subsection{Short-run and long-run equilibria}

\begin{figure}[p]
\begin{minipage} {0.5\columnwidth}
  \centering
  \includegraphics[width=\textwidth]{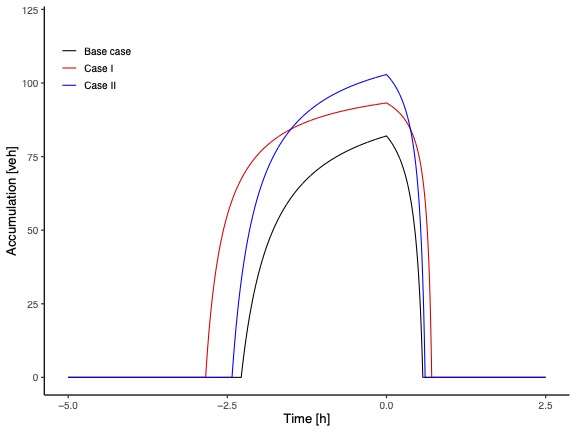}
  \caption{Accumulation in the downtown area}\label{fig:accumulation}
\end{minipage}
\begin{minipage} {0.5\columnwidth}
  \centering
  \includegraphics[width=\textwidth]{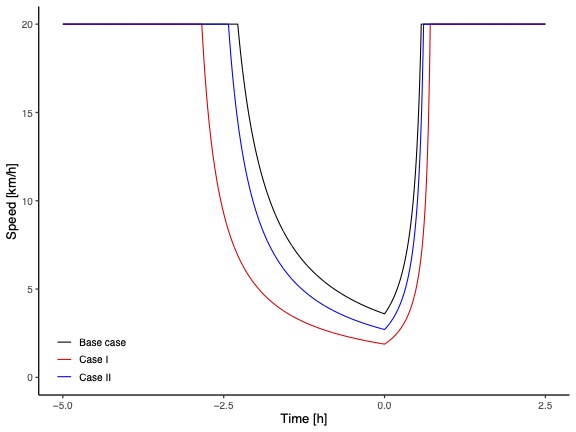}
  \caption{Speed in the downtown area}\label{fig:speed}
\end{minipage}
\begin{minipage} {0.5\columnwidth}
  \centering
  \includegraphics[width=\textwidth]{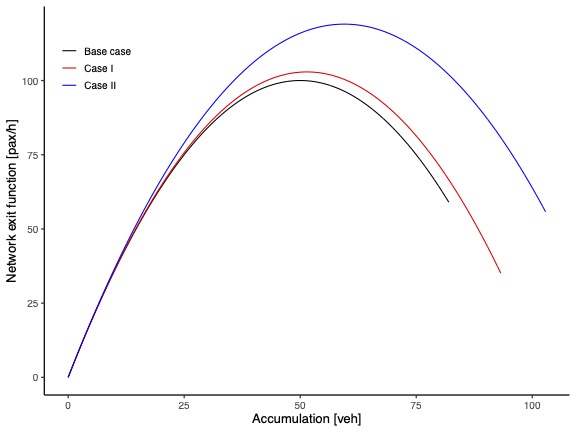}
  \caption{MFD in the downtown area.  }\label{fig:MFD}
\end{minipage}
\begin{minipage} {0.5\columnwidth}
  \centering
  \includegraphics[width=\textwidth]{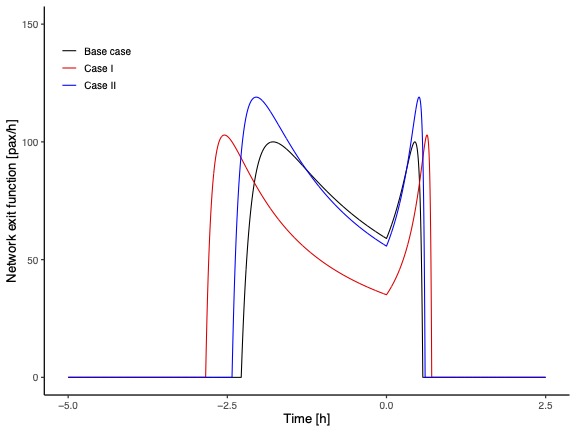}
  \caption{NEF of suburban commuters in downtown area}\label{fig:NEF_UE_AV}
\end{minipage}
\begin{minipage} {0.5\columnwidth}
  \centering
  \includegraphics[width=\textwidth]{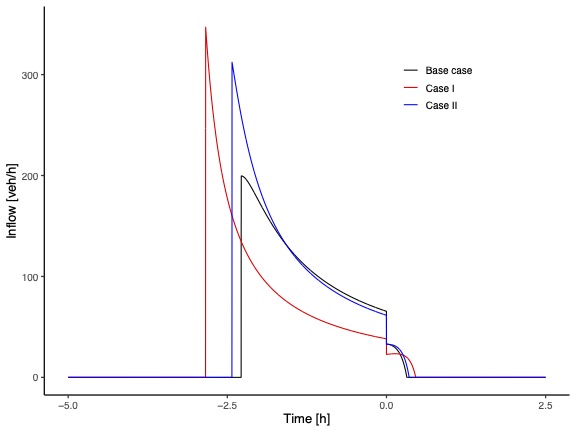}
  \caption{Inflow in downtown area}\label{fig:inflow}
\end{minipage}
\end{figure}

\begin{table}[t]
  \caption{Numerical results of short-run and long-run equilibria.  }
  \label{table:numerical_result}
  \centering
  \begin{tabular}{lcccccc}
    \hline
    Case  &  $N_s$ & $C_s^{b*}$ &$U^*$& $N_s^p$ & $C_s^{bp*}$ & $U^{p*}$ \\
    \hline \hline
    Base case ($\eta=1$, $\xi=1$)  & 224.0 & 27.8 & 4.594 &  252.2  & 26.3 & 4.684 \\
    I ($\eta=0.59$, $\xi=1.029$)  &  221.2 & 31.4 & 4.586 & 305.5 & 27.4 & 4.883 \\
    II ($\eta=0.76$, $\xi=1.19$)  & 256.6 & 28.1 & 4.699 & 308.1 &  25.4 & 4.894 \\
    \hline 
    \multicolumn{7}{r}{{\footnotesize superscript $p$ describes variables under perimeter control.}}
  \end{tabular}
\end{table}

\begin{figure}[!t]
\begin{minipage} {0.5\columnwidth}
  \centering
  \includegraphics[width=\textwidth]{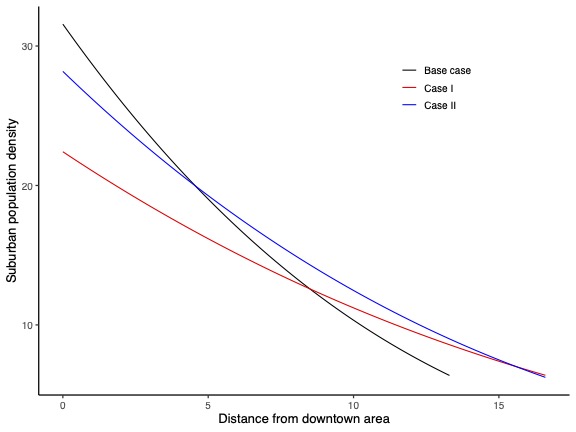}
  \caption{Suburban population density by residential location}\label{fig:suburban_population_density}
\end{minipage}
\begin{minipage} {0.5\columnwidth}
  \centering
  \includegraphics[width=\textwidth]{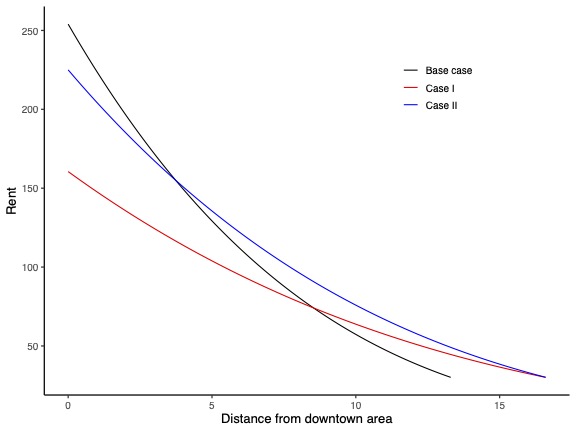}
  \caption{Rent by residential location}\label{fig:Rent_av}
\end{minipage}
\begin{minipage} {0.5\columnwidth}
  \centering
  \includegraphics[width=\textwidth]{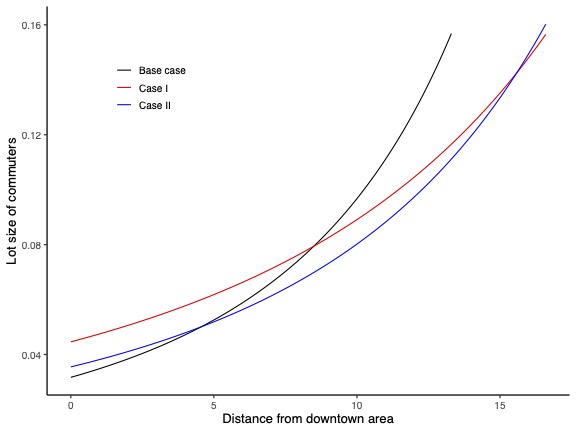}
  \caption{Lot size of commuters by residential location}\label{fig:lot_size_av}
\end{minipage}
\end{figure}

\begin{figure}[p]
\begin{minipage} {0.5\columnwidth}
  \centering
  \includegraphics[width=\textwidth]{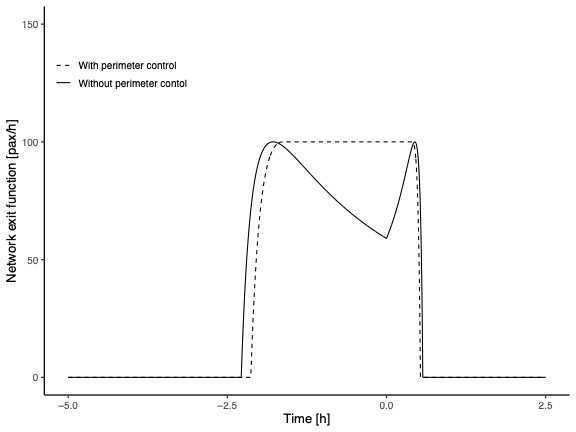}
\subcaption{NEF}\label{fig:NEF_base}
\end{minipage}
\begin{minipage} {0.5\columnwidth}
  \centering
  \includegraphics[width=\textwidth]{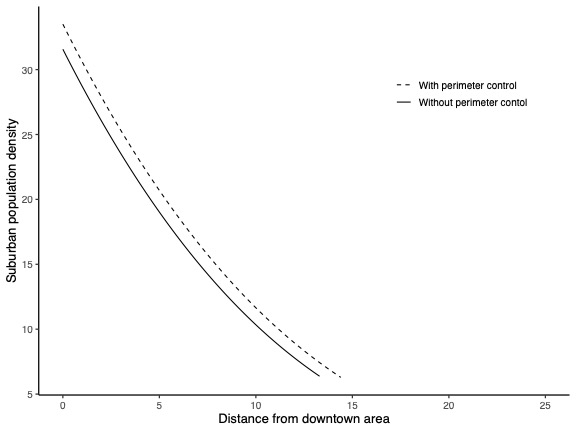}
  \subcaption{Suburban population density}\label{fig:suburban_base}
\end{minipage}
\caption{Base case} \label{fig:base_case}
\begin{minipage} {0.5\columnwidth}
  \centering
  \includegraphics[width=\textwidth]{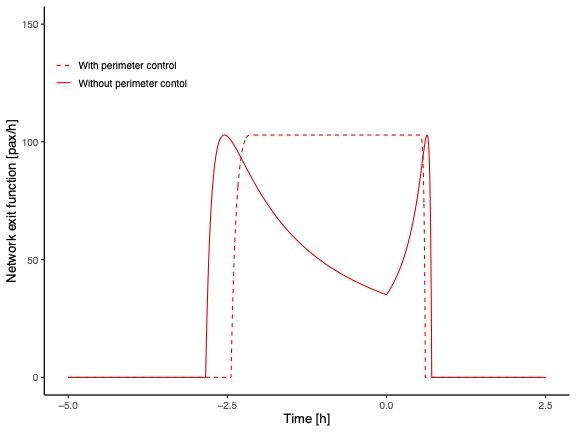}
\subcaption{NEF}\label{fig:NEF_case1}
\end{minipage}
\begin{minipage} {0.5\columnwidth}
  \centering
  \includegraphics[width=\textwidth]{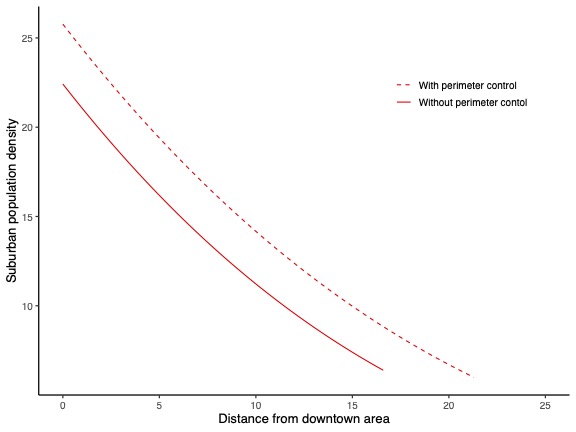}
    \subcaption{Suburban population density}\label{fig:suburban_case1}
\end{minipage}
\caption{Case I} \label{fig:case1}

\begin{minipage} {0.5\columnwidth}
  \centering
  \includegraphics[width=\textwidth]{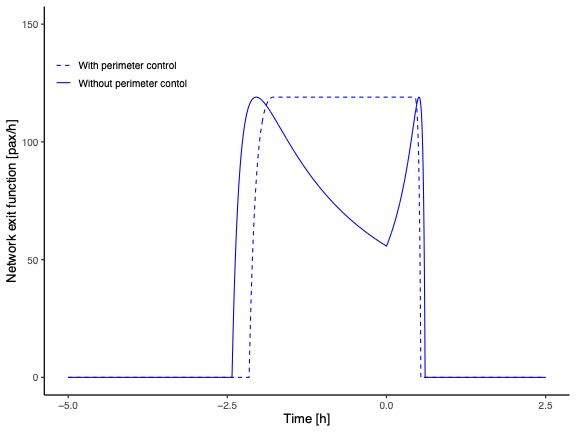}
\subcaption{NEF}\label{fig:NEF_case2}
\end{minipage}
\begin{minipage} {0.5\columnwidth}
  \centering
  \includegraphics[width=\textwidth]{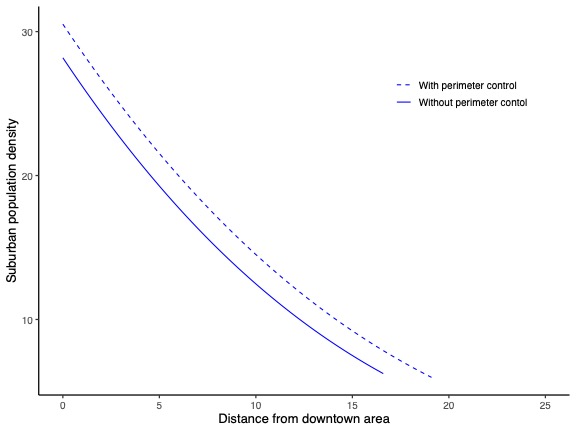}
    \subcaption{Suburban population density}\label{fig:suburban_case2}
\end{minipage}
\caption{Case II} \label{fig:case2}
\end{figure}

We next investigate the equilibrium patterns including the long-run equilibrium. Instead of using the fixed constant suburban commuters in Section \ref{section:short-run}, the total population is set to $N=600$ [pax]. We also set the parameters related to the long-run equilibrium as follows: $r_A = 30$, $w=60$, $\mu=0.25$, $A_d=2$, and $A_s(x)=1$. We use the same cases as in Section \ref{section:short-run} to compare the impact of autonomous vehicles and hypercongestion mitigation by perimeter control on the equilibrium patterns.

Table. \ref{table:numerical_result} shows that the introduction of autonomous vehicles decreases and increases the suburban population in Case I and II, respectively. When autonomous vehicles are introduced, less commuters choose the suburban area as the residential location in the long-run due to an increases in the short-run equilibrium bathtub cost in Case I and more commuters choose the suburban area in the long-run due to a decrease in the short-run equilibrium bathtub cost in Case II. 
Even though the suburban population in Case I is lower than that in the base case, the short-run equilibrium bathtub cost increases (i.e., longer rush hour in Figs. \ref{fig:accumulation} and \ref{fig:speed}) due to a severe capacity drop (Figs. \ref{fig:MFD} and \ref{fig:NEF_UE_AV}) by more temporal concentration of traffic demand by the VOT effect (Fig. \ref{fig:inflow}).  On the other hand, the higher suburban population leads to an increase in the short-run equilibrium bathtub cost in Case II.

{\it Although the suburban population decreases in Case I by the introduction of autonomous vehicles, the city spatially expand outward}, as depicted in Figs. \ref{fig:suburban_population_density} - \ref{fig:lot_size_av}.  This is because as the VOT effect is higher, suburban commuters care less about free-flow travel time in the suburban areas, which results in the farther city boundary. Interestingly, Table \ref{table:numerical_result} revelas that the utility level may decreases despite that commuters can relocate in response to the changes in their commuting behaviors due to the introduction of autonomous vehicles. Therefore, {\it the introduction of autonomous vehicles may decrease the utility level in the long-run due to a severe capacity drop in the short-run.}  

As perimeter control implementation decreases the short-run equilibrium bathtub cost, the suburban population increases, which results in the city expansion, as depicted in Table \ref{table:numerical_result} and Figs. \ref{fig:base_case} - \ref{fig:case2}. Thus, hypercongestion mitigation by perimeter control decreases the short-run equilibrium bathtub cost, and results in a decreases in the downtown population in the long-run.

The short-run equilibrium bathtub cost in Case I  is higher than that in the Base case with and without perimeter control, but the utility level is lower without perimeter control and higher with perimeter control in Case I than those in the Base case.  This is because unlike the case in the presence of hypercongestion, the short-run equilibrium bathtub cost decreases with perimeter control and more  commuters reside in the suburban area in the long-run. Furthermore, the high VOT effect decreases the suburban population density at location $x$. Thus, both the short-run equilibrium bathtub cost and the utility level are higher in Case I than in the Base case. The same effects exist in Case II. However, the high network capacity effect causes more suburban commuters to arrive at their destinations in the short-run, but the low VOT effect causes a higher suburban population density and results in fewer commuters residing in the suburban area in the long-run compared with Case I. Thus, the short-run equilibrium bathtub cost is lower, but the utility level is higher in Case II than in the Base case.

Next, we conduct a sensitivity analysis of utility with respect to both the network capacity and VOT effects, as depicted in Fig.~\ref{fig:sensitivity_utility}. The pictures of the utility levels with and without perimeter control change entirely. Without perimeter control, 
both the network capacity and the VOT effects increase the utility level  when the VOT effect is low (i.e., high value of time in Fig.~\ref{fig:utility_sensitivity_UE}).   However, the utility level begins to decrease when the VOT effect exceeds the threshold values. This is because a high VOT effect leads to a severe capacity drop, resulting in an increase in the short-run equilibrium bathtub cost. Consequently, the income net of the suburban commuting cost decreases and the downtown population increases, leading to a lower utility level in the long-run.  
When the VOT effect is low, the short-run equilibrium bathtub cost decreases, and the free-flow travel time cost in the suburban area also decreases, resulting in a higher utility level. When perimeter control is implemented, both network capacity and VOT effects always increase the utility level, as depicted in Fig.~\ref{fig:utility_sensitivity_perimeter_control}. As discussed earlier, both effects consistently decrease the short-run equilibrium bathtub cost, resulting in an increase in income net of the suburban commuting cost and a decrease in  the downtown population. They contribute to a higher utility level.

\begin{figure}[!t]
\begin{minipage} {0.5\columnwidth}
  \centering
  \includegraphics[width=\textwidth]{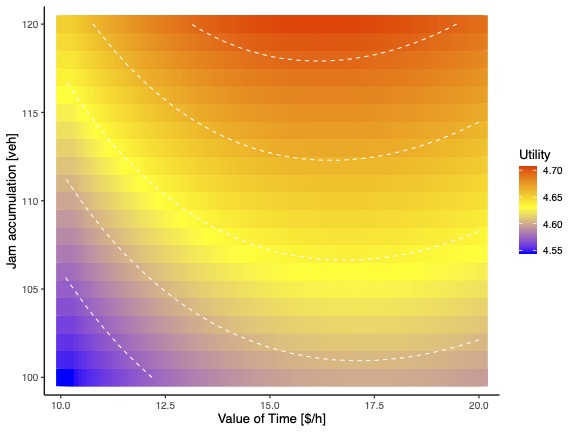}
  \subcaption{Without perimeter control}\label{fig:utility_sensitivity_UE}
\end{minipage}
\begin{minipage} {0.5\columnwidth}
  \centering
  \includegraphics[width=\textwidth]{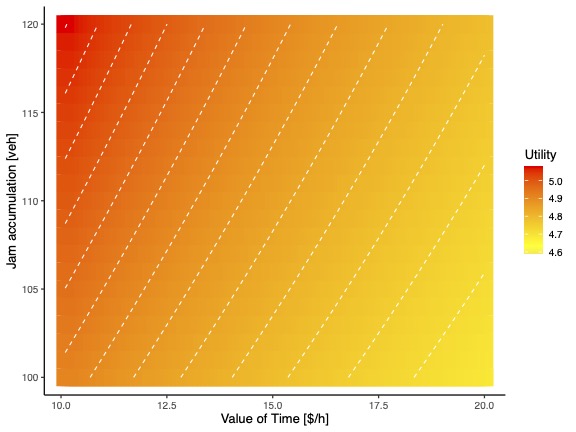}
  \subcaption{With perimeter control}\label{fig:utility_sensitivity_perimeter_control}
\end{minipage}
\caption{Sensitivity analysis of utility}
\label{fig:sensitivity_utility}
\end{figure}

\section{Conclusions}
In this paper, by  incorporating a bathtub model, we develop a land use model where hypercongestion occurs in the downtown area and interacts with land use  to examine the effects of hypercongestion mitigation by perimeter control and the introduction of autonomous vehicles on the spatial structures of cities. The results indicate that (I) hypercongestion mitigation decreases the commuting cost and results in a less dense urban spatial structure, (II) the introduction of autonomous vehicles may increase the commuting cost in the presence of hypercongestion and results in a decrease in the suburban population (i.e., total car traffic demand) in the long-run, but make cities spatially expanded outward,  and (III) the introduction of autonomous vehicles under perimeter control decreases the commuting cost and results in a less dense urban spatial structure. We also show that the introduction of autonomous vehicles may decrease the utility level in the long-run due to a severe capacity drop in the short-run. These results show that hypercongestion is a key factor that can change urban spatial structure. Specifically, the effect of autonomous vehicles on the urban spatial structures depends on  the presence of hypercongestion. Moreover, the effect of autonomous vehicles on the urban spatial structure in the presence of hypercongestion contradicts that in the standard bottleneck model. 

This paper has several future directions. First, we assumed commuter homogeneity. Incorporating heterogeneity such as in trip length \cite[e.g.,][]{fosgerau2015congestion, lamotte2018morning} and preferences \cite[e.g.,][]{hall2018pareto, takayama2020scheduling} is an  interesting direction. Particularly, since changes in the OD demand patterns in the downtown area can impact the shape of the MFD, incorporating the city structure of the downtown area and the trip-based bathtub model is worth investigating.  Second, we assumed a unimodal transportation system. Future models can be extended to multimodal transportation systems \cite[e.g.,][]{dantsuji2022perimeter}. Third, we examined a scenario where all human-driven vehicles are replaced by autonomous vehicles. Extending it to scenarios involving mixed traffic of autonomous and human-driven vehicles would be desirable. Although there are some studies on the mixed traffic of autonomous and human-driven vehicles in the standard bottleneck model with heterogeneous commuters \cite[e.g.,][]{van2016autonomous, wu2023managing}, developing a mixed bimodal bathtub model with heterogeneous commuters
in terms of value of travel time remains a challenging task, which represents one of the important future directions. Finally, the  monocentric city structure in this work can be extended to a  polycentric city structure and/or multi-region traffic systems. 

\section*{Acknowledgements}
We thank Se-il Mun, Tatsuhiko Kono, Minoru Osawa, and Koki Satsukawa for their valuable comments. This work was supported by JST ACT-X, Japan (grant \#JPMJAX21AE), by JST FOREST Program, Japan (grant \#JPMJFR215M), by JSPS KAKENHI, Japan (grant \#23K13422 and \#22H01610).

\section*{Appendix A: proof of Lemma 1}
As \cite{small2003hypercongestion} assumed the Ardekani--Herman formula (i.e., $v(t)=v_f \left( 1 - n(t)/n_j\right)^{1+\rho}$, where $\rho$ is a parameter) for the space-mean speed in the downtown area, the Greenshields model in this paper can be regarded as a special case of the \cite{small2003hypercongestion}  model where $\rho=0$. Here, we briefly review the solution of the bathtub model. Differentiating the bathtub  cost $C_s^b(t)=\alpha T(t) + s(t)$ with respect to time at equilibrium (i.e., ${\mathrm{d} C_s^b(t)}/{\mathrm{d}t}=0$) yields 
\begin{align} \label{eq:differential_bathtub_cost}
\frac{\mathrm{d} T(t)}{\mathrm{d}t} =
\begin{cases}
\frac{\beta}{\alpha} & \qquad  \text{if } t \leq t^*  \\
- \frac{\gamma}{\alpha} & \qquad \text{if }t > t^*.
\end{cases}
\end{align}
From Eqs.~(\ref{eq:speed}) and (\ref{eq:TT_assumption}), we have 
\begin{align}\label{eq:differential_TT}
  \frac{\mathrm{d} T(t)}{\mathrm{d}t} = \frac{L}{n_j v_f} \left( 1 - \frac{n(t)}{n_j}\right)^{-2}.
\end{align}
Combining Eqs.~(\ref{eq:differential_bathtub_cost}) and (\ref{eq:differential_TT}),  grouping the terms in $n(t)$ and the other terms on the LHS and RHS, respectively, and integrating both sides yields 
\begin{align} \label{eq:accumulation}
n(t) =
\begin{cases}
n_j \left( 1 - \frac{1}{\frac{\beta v_f}{\alpha L}t + C_e}\right) & \qquad  \text{if } t \leq t^*  \\
n_j  \left( 1 - \frac{1}{-\frac{\gamma v_f}{\alpha L}t + C_l}\right) & \qquad \text{if }t > t^*,
\end{cases}
\end{align}
where $C_e$ and $C_l$ are constants of integration for earliness and lateness, respectively. Since the accumulation at the start and end of the rush hour is zero (i.e., $n(t_s)=n(t_e)=0$),
\begin{align} \label{eq:accumulation_derived}
n(t) =
\begin{cases}
n_j \left( 1 - \frac{1}{1 + \frac{\beta v_f}{\alpha L} (t - t_s) }\right) & \qquad  \text{if } t \leq t^*  \\
n_j  \left( 1 - \frac{1}{1 + \frac{\gamma v_f}{\alpha L} (t_e - t) }\right) & \qquad \text{if }t > t^*,
\end{cases}
\end{align}
From Condition (\ref{eq:N_s}), 
\begin{align} 
 N_s = \alpha n_j \left( \frac{1}{\beta} + \frac{1}{\gamma} \right) \left( \log \theta + \frac{1}{\theta} - 1 \right),
\end{align} 
where $\theta = \frac{C_s^{b*} v_f}{\alpha L }$. 

When $\theta=2$, the maximum accumulation reached during the rush hour is the critical accumulation ($n_j/2$) at the desired arrival time ($n(t^*)=n_j/2$). Therefore, hypercongestion exists if $\theta>2$. 

\section*{Appendix B: proof of Lemma \ref{lemma:long_run_equilibrium_property0}}
As $\frac{d U_s(x)}{dx}=0$ at long-run equilibrium, we have from Eq.~(\ref{eq:indirect_utility_suburban}), 
\begin{align} \label{eq:dr_dx}
\frac{d r_s(x)+r_A}{dx} = - \alpha \tau \frac{ r_s(x)+r_A}{\mu y_s(x)}.
\end{align}
The RHS is negative. 

Combining Condition~(\ref{eq:long4}) and Eq.~(\ref{eq:lot_size_equilibrium}) yields 
\begin{align} \label{eq:ns_as}
\frac{N_s(x)}{A_s(x)} = \frac{r_s(x)+r_A}{\mu y_s(x)}.
\end{align}
By differentiating Eq.~(\ref{eq:ns_as}) with respect to $x$ and substituting Eq.~(\ref{eq:dr_dx}) into Eq.~(\ref{eq:ns_as}), we obtain 
\begin{align}
\frac{d \frac{N_s(x)}{A_s(x)}}{dx} = - \alpha \tau \frac{(1-\mu) (r_s(x)+r_A)}{ \{ \mu \left(w - C_s^{b*} - \alpha \tau x\right)\}^2}.
\end{align}
The RHS is negative.

\section*{Appendix C: proof  of Lemma \ref{lemma:long_run_equilibrium_property}}
At the long-run equilibrium, the indirect utility satisfies $U^*_d=U_s^*(x)$ for all $ x \in [ 0, x_f ] $ , and this condition   gives 
\begin{align}  \label{eq:income_rent_relation}
\frac{r_d+r_A}{r_s(x)+r_A} = \left(\frac{y_d}{y_s(x)}\right)^{\frac{1}{\mu}}.
\end{align} 
As we have $a_d=\frac{A_d}{N_d}$ from Condition (\ref{eq:long3}), substituting these into Eq.~(\ref{eq:lot_size_equilibrium}) results in $N^*_d$ as follows. 
\begin{align} 
N^*_d &=\frac{1}{\mu} \{y_d\}^{\frac{1-\mu}{\mu}} \left( w - C_s^{b*} \right)^{-\frac{1}{\mu}} \left( r_s(0) + r_A \right)  A_d.
\end{align} 
Similarly, we have  
\begin{align} 
&N_s(x) =  \nonumber \\
& \frac{1}{\mu} \{w - C_s^{b*} - \alpha \tau x \}^{\frac{1-\mu}{\mu}} \{  y_d \}^{-\frac{1}{\mu}} \left( r_s(0) + r_A \right)  A_s(x). 
\end{align} 
Furthermore, as $U_s^*(0)=U_s^*(x_f)$, we obtain 
\begin{align} 
  x_f = \frac{w-C_s^{b*}}{\alpha \tau} \left( \{ r_A \}^{-\mu} - \{ r_s(0) + r_A \}^{-\mu} \right)\{ r_A \}^{\mu}.
\end{align}

\section*{Appendix D: proof of Lemma \ref{lemma:short-run_perimeter}}
As the NEF is maintained at the maximum during perimeter control, the number of suburban commuters who arrive at their destinations during perimeter control is 
\begin{align}  \label{eq:number_commuters_during_perimeter}
N_s^p =  \frac{n_jv_f}{4L} \left(t_e^p - t_s^p \right).
\end{align} 
Since $C_s^{p*}= 2 \alpha L/v_f + \beta (t^* - t_s^p) =  2 \alpha L/v_f + \beta ( t_e^p - t^*) $, substituting it into Eq.~(\ref{eq:number_commuters_during_perimeter}) produces 
\begin{align}  \label{eq:number_commuters_during_perimeterv2}
N_s^p =  \frac{\alpha n_j}{4} \left( \frac{1}{\beta} + \frac{1}{\gamma} \right) \left( \theta^p - 2 \right).
\end{align} 
According to Eq.~(\ref{eq:accumulation_derived}), the start and end times of perimeter control ($t_s^p$ and $t_e^p$, respectively; $n(t_s^p)=n(t_e^p)=n_j/2$) are
\begin{align}  
t_s^p - t_s = \frac{1}{\beta} \frac{\alpha L}{v_f}, \label{eq:earliness_perimeter}\\
t_e - t_e^p = \frac{1}{\gamma} \frac{\alpha L}{v_f}.  \label{eq:lateness_perimeter}
\end{align} 
Therefore, the number of suburban commuters who arrive at their destinations before and after perimeter control is computed using Eqs.~(\ref{eq:speed}), (\ref{eq:accumulation_derived}), (\ref{eq:earliness_perimeter}), and (\ref{eq:lateness_perimeter}) as follows: 
\begin{align}  \label{eq:number_commuters_outside_perimeter}
N_s^{op} &=  \int_{t_s}^{t_s^p} \frac{n(s)v(s)}{L} ds + \int_{t_e^p}^{t_e} \frac{n(s)v(s)}{L} ds  \nonumber \\
&= \alpha n_j \left( \frac{1}{\beta} + \frac{1}{\gamma} \right)  \left( \ln 2 - \frac{1}{2} \right).
\end{align} 
We combine Eqs.~(\ref{eq:number_commuters_during_perimeterv2}) and (\ref{eq:number_commuters_outside_perimeter}), and the number of suburban commuters is expressed as 
\begin{align} \label{eq:number_commuters_perimeter}
        N_s = \alpha n_j \left( \frac{1}{\beta} + \frac{1}{\gamma} \right) \left( \frac{\theta^p}{4} + \ln 2 - 1 \right).
\end{align}

Next, we prove that a queue develops at the perimeter boundary, and its length increases toward the desired arrival time. Eq. (\ref{eq:gc_perimeter}) and Condition (\ref{eq:UE_con_car_perimeter}) yield 
\begin{align}
\begin{cases}
\alpha   \frac{\mathrm{d} T_w(t)}{\mathrm{d} t } - \beta = 0  \\
\\
\alpha   \frac{\mathrm{d} T_w(t)}{\mathrm{d} t } + \gamma = 0.
\end{cases}    
\end{align}
Combining this with Eq.~(\ref{eq:waiting_time_boundary}), where $I_p$ is constant, we obtain 
\begin{align}
\begin{cases}
  \frac{\mathrm{d} q(t)}{\mathrm{d} t } = \frac{\beta}{\alpha}  \\
\\
\frac{\mathrm{d} q(t)}{\mathrm{d} t } = - \frac{\gamma}{\alpha}.
\end{cases}    
\end{align}
Thus, a queue starts to develop once perimeter control is implemented, and its length increases toward the desired arrival time. Then its length decreases after the desired arrival time.

\section*{Appendix E: proof of $C_s^{b*} > C_s^{pb*}$}
When hypercongestion exists without perimeter control, $\theta > 2$ from Lemma \ref{lemma:short_run_without}.  Consider a set of parameters (including $N_s$), $\Psi$, where there exists $\theta(>2)$ such that $F(\theta)=0$. Under $\Psi$, it always holds that $F^p(\theta)<0$. This leads to  $\theta > \theta^p$ under $\Psi$ if there exists $\theta^p(>2)$  such that $F^p(\theta^p)=0$. $\theta > \theta^p$ yields $C_s^{b*} > C_s^{pb*}$, which completes the proof.

\section*{Appendix F: biases in the critical accumulation estimation }
Since the controlled inflow could be influenced by various factors such as infrastructure limitation \citep{aboudolas2013perimeter}, the selection of  measurement points for the MFD estimation \citep{ortigosa2015study}, and the estimation methods \citep{leclercq2014macroscopic}, it is important to investigate the effects of these biases. To this end, we introduce the bias factor into the critical accumulation. 
\begin{align} \label{eq:epsilon_accumulation}
    n_{cr} = \frac{\epsilon n_j}{2},
\end{align}
where $\epsilon$ ($>0$) is the bias factor caused by the reasons mentioned above. $\epsilon>1$ indicates the accumulation under perimeter control is above the critical value.

As the only difference from the non-biased perimeter control is the bias factor $\epsilon$, in the same way as Section 4, we have the short-run equilibrium bathtub cost satisfying
\begin{align}
& F_{\epsilon}^p(\theta_{\epsilon}^p) \equiv  \nonumber \\ & N_s - \alpha n_j \left( \frac{1}{\beta} + \frac{1}{\gamma} \right) \left( \frac{\epsilon (2 - \epsilon)}{4} \theta_{\epsilon}^p + \ln \frac{2}{2-\epsilon}  - \epsilon \right)  \nonumber \\
&  = 0. \label{eq:equilibrium_bathtub_cost_perimeter_epsilon}
\end{align}
Note that if the critical accumulation estimation is not biased (i.e., $\epsilon=1$), Eq. (\ref{eq:equilibrium_bathtub_cost_perimeter_epsilon}) is identical to Eq. (\ref{eq:equilibrium_bathtub_cost_perimeter}).

We then conduct the numerical analyses for two cases: (A) higher accumulation than the critical value ($\epsilon=1.3$), and (B) lower accumulation than the critical value ($\epsilon=0.7$). Note that the outflow under perimeter control is maintained at the same level for both cases, but the traffic state is hypercongested in Case (A) and congested in Case (B). Fig. \ref{fig:epsilon_NEF} displays the NEFs for the cases including the non-biased case ($\epsilon=1$) and the case without perimeter control (UE).

We observe that the rush hour of the non-biased case is the shortest because the outflow can be maintained at its maximum, leading to the lowest short-run bathtub equilibrium cost. Interestingly, the rush hour in Case (A) is shorter than that in Case (B), even though the accumulation under perimeter control is in a hypercongested state for Case (A). This is because the maximum outflow is never reached for Case (B), while there are time windows in Case (A) when the outflow is higher than the outflow under perimeter control.

Since the traffic dynamics is unstable in the hypercongested state, further investigation is required. However, the obtained result suggests the possibility that controlling the inflow such that the outflow is slightly below the optimal value may be less effective than controlling the inflow such that the outflow is slightly above the optimal value. 

\begin{figure}
\centering
 \includegraphics[width=0.5\textwidth]{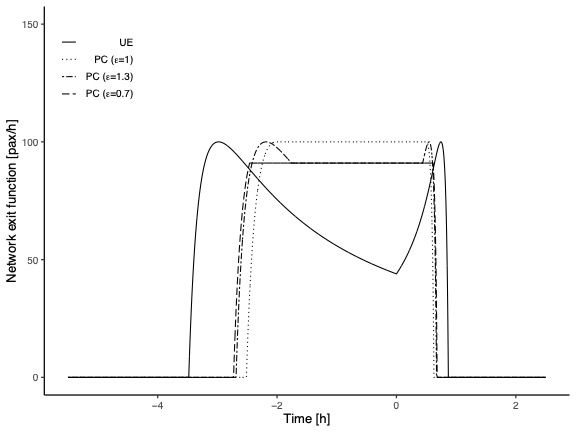}
  \caption{Impact of bias in the critical accumulaion estimation}\label{fig:epsilon_NEF}
\end{figure}

\bibliographystyle{apalike}
\bibliography{reference}

\begin{thebibliography}{}

\bibitem[Aboudolas and Geroliminis, 2013]{aboudolas2013perimeter}
Aboudolas, K. and Geroliminis, N. (2013).
\newblock Perimeter and boundary flow control in multi-reservoir heterogeneous
  networks.
\newblock {\em Transportation Research Part B: Methodological}, 55:265--281.

\bibitem[Alonso, 1964]{alonso1964historic}
Alonso, W. (1964).
\newblock Location and land use.
\newblock {\em Havard university press}.

\bibitem[Ameli et~al., 2022]{ameli2022departure}
Ameli, M., Faradonbeh, M. S.~S., Lebacque, J.-P., Abouee-Mehrizi, H., and
  Leclercq, L. (2022).
\newblock Departure time choice models in urban transportation systems based on
  mean field games.
\newblock {\em Transportation Science}, 56(6):1483--1504.

\bibitem[Ampountolas et~al., 2017]{ampountolas2017macroscopic}
Ampountolas, K., Zheng, N., and Geroliminis, N. (2017).
\newblock Macroscopic modelling and robust control of bi-modal multi-region
  urban road networks.
\newblock {\em Transportation Research Part B: Methodological}, 104:616--637.

\bibitem[Anas et~al., 1998]{anas1998urban}
Anas, A., Arnott, R., and Small, K.~A. (1998).
\newblock Urban spatial structure.
\newblock {\em Journal of economic literature}, 36(3):1426--1464.

\bibitem[Arnott, 1998]{arnott1998congestion}
Arnott, R. (1998).
\newblock Congestion tolling and urban spatial structure.
\newblock {\em Journal of regional science}, 38(3):495--504.

\bibitem[Arnott, 2013]{arnott2013bathtub}
Arnott, R. (2013).
\newblock A bathtub model of downtown traffic congestion.
\newblock {\em Journal of Urban Economics}, 76:110--121.

\bibitem[Arnott et~al., 1993]{arnott1993structural}
Arnott, R., de~Palma, A., and Lindsey, R. (1993).
\newblock A structural model of peak-period congestion: A traffic bottleneck
  with elastic demand.
\newblock {\em The American Economic Review}, pages 161--179.

\bibitem[Bao et~al., 2021]{bao2021leaving}
Bao, Y., Verhoef, E.~T., and Koster, P. (2021).
\newblock Leaving the tub: The nature and dynamics of hypercongestion in a
  bathtub model with a restricted downstream exit.
\newblock {\em Transportation Research Part E: Logistics and Transportation
  Review}, 152:102389.

\bibitem[Batista et~al., 2021]{batista2021role}
Batista, S.~F., Ingole, D., Leclercq, L., and Men{\'e}ndez, M. (2021).
\newblock The role of trip lengths calibration in model-based perimeter control
  strategies.
\newblock {\em IEEE Transactions on Intelligent Transportation Systems},
  23(6):5176--5186.

\bibitem[Chen et~al., 2022a]{chen2022data}
Chen, C., Huang, Y., Lam, W., Pan, T., Hsu, S., Sumalee, A., and Zhong, R.
  (2022a).
\newblock Data efficient reinforcement learning and adaptive optimal perimeter
  control of network traffic dynamics.
\newblock {\em Transportation Research Part C: Emerging Technologies},
  142:103759.

\bibitem[Chen et~al., 2022b]{chen2022passenger}
Chen, S., Fu, H., Wu, N., Wang, Y., and Qiao, Y. (2022b).
\newblock Passenger-oriented traffic management integrating perimeter control
  and regional bus service frequency setting using 3d-pmfd.
\newblock {\em Transportation Research Part C: Emerging Technologies},
  135:103529.

\bibitem[Chiabaut, 2015]{chiabaut2015evaluation}
Chiabaut, N. (2015).
\newblock Evaluation of a multimodal urban arterial: The passenger macroscopic
  fundamental diagram.
\newblock {\em Transportation Research Part B: Methodological}, 81:410--420.

\bibitem[Daganzo, 2007]{daganzo2007urban}
Daganzo, C.~F. (2007).
\newblock Urban gridlock: Macroscopic modeling and mitigation approaches.
\newblock {\em Transportation Research Part B: Methodological}, 41(1):49--62.

\bibitem[Dantsuji et~al., 2021]{dantsuji2021simulation}
Dantsuji, T., Fukuda, D., and Zheng, N. (2021).
\newblock Simulation-based joint optimization framework for congestion
  mitigation in multimodal urban network: a macroscopic approach.
\newblock {\em Transportation}, 48(2):673--697.

\bibitem[Dantsuji et~al., 2020]{dantsuji2019cross}
Dantsuji, T., Hirabayashi, S., Ge, Q., and Fukuda, D. (2020).
\newblock Cross comparison of spatial partitioning methods for an urban
  transportation network.
\newblock {\em International Journal of Intelligent Transportation Systems
  Research}, 18:412--421.

\bibitem[Dantsuji et~al., 2022]{dantsuji2022novel}
Dantsuji, T., Hoang, N.~H., Zheng, N., and Vu, H.~L. (2022).
\newblock A novel metamodel-based framework for large-scale dynamic
  origin--destination demand calibration.
\newblock {\em Transportation Research Part C: Emerging Technologies},
  136:103545.

\bibitem[Dantsuji et~al., 2023]{dantsuji2022perimeter}
Dantsuji, T., Takayama, Y., and Fukuda, D. (2023).
\newblock Perimeter control in a mixed bimodal bathtub model.
\newblock {\em Transportation Research Part B: Methodological}, 173:267--291.

\bibitem[de~Almeida~Correia et~al., 2019]{de2019impact}
de~Almeida~Correia, G.~H., Looff, E., van Cranenburgh, S., Snelder, M., and van
  Arem, B. (2019).
\newblock On the impact of vehicle automation on the value of travel time while
  performing work and leisure activities in a car: Theoretical insights and
  results from a stated preference survey.
\newblock {\em Transportation Research Part A: Policy and Practice},
  119:359--382.

\bibitem[Ding et~al., 2017]{ding2017traffic}
Ding, H., Guo, F., Zheng, X., and Zhang, W. (2017).
\newblock Traffic guidance--perimeter control coupled method for the congestion
  in a macro network.
\newblock {\em Transportation Research Part C: Emerging Technologies},
  81:300--316.

\bibitem[Fagnant and Kockelman, 2015]{fagnant2015preparing}
Fagnant, D.~J. and Kockelman, K. (2015).
\newblock Preparing a nation for autonomous vehicles: opportunities, barriers
  and policy recommendations.
\newblock {\em Transportation Research Part A: Policy and Practice},
  77:167--181.

\bibitem[Fosgerau, 2015]{fosgerau2015congestion}
Fosgerau, M. (2015).
\newblock Congestion in the bathtub.
\newblock {\em Economics of Transportation}, 4(4):241--255.

\bibitem[Fosgerau and Kim, 2019]{fosgerau2019commuting}
Fosgerau, M. and Kim, J. (2019).
\newblock Commuting and land use in a city with bottlenecks: Theory and
  evidence.
\newblock {\em Regional Science and Urban Economics}, 77:182--204.

\bibitem[Fosgerau et~al., 2018]{fosgerau2018vickrey}
Fosgerau, M., Kim, J., and Ranjan, A. (2018).
\newblock Vickrey meets alonso: Commute scheduling and congestion in a
  monocentric city.
\newblock {\em Journal of Urban Economics}, 105:40--53.

\bibitem[Fosgerau and Small, 2013]{fosgerau2013hypercongestion}
Fosgerau, M. and Small, K.~A. (2013).
\newblock Hypercongestion in downtown metropolis.
\newblock {\em Journal of Urban Economics}, 76:122--134.

\bibitem[Fu et~al., 2021]{fu2021perimeter}
Fu, H., Chen, S., Chen, K., Kouvelas, A., and Geroliminis, N. (2021).
\newblock Perimeter control and route guidance of multi-region mfd systems with
  boundary queues using colored petri nets.
\newblock {\em IEEE Transactions on Intelligent Transportation Systems},
  23(8):12977--12999.

\bibitem[Fujita, 1989]{fujita1989urban}
Fujita, M. (1989).
\newblock Urban economic theory.
\newblock {\em Cambridge Books}.

\bibitem[Genser and Kouvelas, 2022]{genser2022dynamic}
Genser, A. and Kouvelas, A. (2022).
\newblock Dynamic optimal congestion pricing in multi-region urban networks by
  application of a multi-layer-neural network.
\newblock {\em Transportation Research Part C: Emerging Technologies},
  134:103485.

\bibitem[Geroliminis, 2015]{geroliminis2015cruising}
Geroliminis, N. (2015).
\newblock Cruising-for-parking in congested cities with an mfd representation.
\newblock {\em Economics of Transportation}, 4(3):156--165.

\bibitem[Geroliminis and Daganzo, 2008]{geroliminis2008existence}
Geroliminis, N. and Daganzo, C.~F. (2008).
\newblock Existence of urban-scale macroscopic fundamental diagrams: Some
  experimental findings.
\newblock {\em Transportation Research Part B: Methodological}, 42(9):759--770.

\bibitem[Geroliminis et~al., 2012]{geroliminis2012optimal}
Geroliminis, N., Haddad, J., and Ramezani, M. (2012).
\newblock Optimal perimeter control for two urban regions with macroscopic
  fundamental diagrams: A model predictive approach.
\newblock {\em IEEE Transactions on Intelligent Transportation Systems},
  14(1):348--359.

\bibitem[Geroliminis and Levinson, 2009]{geroliminis2009cordon}
Geroliminis, N. and Levinson, D.~M. (2009).
\newblock Cordon pricing consistent with the physics of overcrowding.
\newblock In {\em Transportation and Traffic Theory 2009: Golden Jubilee},
  pages 219--240. Springer.

\bibitem[Godfrey, 1969]{godfrey1969mechanism}
Godfrey, J. (1969).
\newblock The mechanism of a road network.
\newblock {\em Traffic Engineering \& Control}, 8(8).

\bibitem[Gonzales, 2015]{gonzales2015coordinated}
Gonzales, E.~J. (2015).
\newblock Coordinated pricing for cars and transit in cities with
  hypercongestion.
\newblock {\em Economics of Transportation}, 4(1-2):64--81.

\bibitem[Gonzales and Daganzo, 2012]{gonzales2012morning}
Gonzales, E.~J. and Daganzo, C.~F. (2012).
\newblock Morning commute with competing modes and distributed demand: user
  equilibrium, system optimum, and pricing.
\newblock {\em Transportation Research Part B: Methodological},
  46(10):1519--1534.

\bibitem[Gubins and Verhoef, 2014]{gubins2014dynamic}
Gubins, S. and Verhoef, E.~T. (2014).
\newblock Dynamic bottleneck congestion and residential land use in the
  monocentric city.
\newblock {\em Journal of Urban Economics}, 80:51--61.

\bibitem[Guo and Ban, 2020]{guo2020macroscopic}
Guo, Q. and Ban, X.~J. (2020).
\newblock Macroscopic fundamental diagram based perimeter control considering
  dynamic user equilibrium.
\newblock {\em Transportation Research Part B: Methodological}, 136:87--109.

\bibitem[Haddad, 2017]{haddad2017optimal}
Haddad, J. (2017).
\newblock Optimal perimeter control synthesis for two urban regions with
  aggregate boundary queue dynamics.
\newblock {\em Transportation Research Part B: Methodological}, 96:1--25.

\bibitem[Haddad and Geroliminis, 2012]{haddad2012stability}
Haddad, J. and Geroliminis, N. (2012).
\newblock On the stability of traffic perimeter control in two-region urban
  cities.
\newblock {\em Transportation Research Part B: Methodological},
  46(9):1159--1176.

\bibitem[Haddad and Shraiber, 2014]{haddad2014robust}
Haddad, J. and Shraiber, A. (2014).
\newblock Robust perimeter control design for an urban region.
\newblock {\em Transportation Research Part B: Methodological}, 68:315--332.

\bibitem[Haddad and Zheng, 2020]{haddad2020adaptive}
Haddad, J. and Zheng, Z. (2020).
\newblock Adaptive perimeter control for multi-region accumulation-based models
  with state delays.
\newblock {\em Transportation Research Part B: Methodological}, 137:133--153.

\bibitem[Haitao et~al., 2019]{haitao2019providing}
Haitao, H., Yang, K., Liang, H., Menendez, M., and Guler, S.~I. (2019).
\newblock Providing public transport priority in the perimeter of urban
  networks: A bimodal strategy.
\newblock {\em Transportation Research Part C: Emerging Technologies},
  107:171--192.

\bibitem[Hall, 2018]{hall2018pareto}
Hall, J.~D. (2018).
\newblock Pareto improvements from lexus lanes: The effects of pricing a
  portion of the lanes on congested highways.
\newblock {\em Journal of Public Economics}, 158:113--125.

\bibitem[Huang et~al., 2023]{huang2023characterizing}
Huang, Y., Ye, Y., Sun, J., and Tian, Y. (2023).
\newblock Characterizing the impact of autonomous vehicles on macroscopic
  fundamental diagrams.
\newblock {\em IEEE Transactions on Intelligent Transportation Systems}.

\bibitem[Ji and Geroliminis, 2012]{ji2012spatial}
Ji, Y. and Geroliminis, N. (2012).
\newblock On the spatial partitioning of urban transportation networks.
\newblock {\em Transportation Research Part B: Methodological},
  46(10):1639--1656.

\bibitem[Jin, 2020]{jin2020generalized}
Jin, W.-L. (2020).
\newblock Generalized bathtub model of network trip flows.
\newblock {\em Transportation Research Part B: Methodological}, 136:138--157.

\bibitem[Kanemoto, 1980]{kanemoto1980theories}
Kanemoto, Y. (1980).
\newblock Theories of urban externalities.
\newblock {\em North-Holland}.

\bibitem[Kolarova and Cherchi, 2021]{kolarova2021impact}
Kolarova, V. and Cherchi, E. (2021).
\newblock Impact of trust and travel experiences on the value of travel time
  savings for autonomous driving.
\newblock {\em Transportation Research Part C: Emerging Technologies},
  131:103354.

\bibitem[Kouvelas et~al., 2023]{kouvelas2023linear}
Kouvelas, A., Saeedmanesh, M., and Geroliminis, N. (2023).
\newblock A linear-parameter-varying formulation for model predictive perimeter
  control in multi-region mfd urban networks.
\newblock {\em Transportation Science}.

\bibitem[Lamotte et~al., 2017]{lamotte2017use}
Lamotte, R., De~Palma, A., and Geroliminis, N. (2017).
\newblock On the use of reservation-based autonomous vehicles for demand
  management.
\newblock {\em Transportation Research Part B: Methodological}, 99:205--227.

\bibitem[Lamotte and Geroliminis, 2018]{lamotte2018morning}
Lamotte, R. and Geroliminis, N. (2018).
\newblock The morning commute in urban areas with heterogeneous trip lengths.
\newblock {\em Transportation Research Part B: Methodological}, 117:794--810.

\bibitem[Leclercq et~al., 2014]{leclercq2014macroscopic}
Leclercq, L., Chiabaut, N., and Trinquier, B. (2014).
\newblock Macroscopic fundamental diagrams: A cross-comparison of estimation
  methods.
\newblock {\em Transportation Research Part B: Methodological}, 62:1--12.

\bibitem[Li et~al., 2021]{li2021perimeter}
Li, Y., Mohajerpoor, R., and Ramezani, M. (2021).
\newblock Perimeter control with real-time location-varying cordon.
\newblock {\em Transportation Research Part B: Methodological}, 150:101--120.

\bibitem[Li et~al., 2023]{lianalytical}
Li, Z.-C., Yu, D.-P., and de~Palma, A. (2023).
\newblock An analytical model for residential location choices of heterogeneous
  households in a monocentric city with stochastic bottleneck congestion.
\newblock {\em THEMA Working Papers}, 2023-01.

\bibitem[Liu and Geroliminis, 2016]{liu2016modeling}
Liu, W. and Geroliminis, N. (2016).
\newblock Modeling the morning commute for urban networks with
  cruising-for-parking: An mfd approach.
\newblock {\em Transportation Research Part B: Methodological}, 93:470--494.

\bibitem[Loder et~al., 2019]{loder2019understanding}
Loder, A., Amb{\"u}hl, L., Menendez, M., and Axhausen, K.~W. (2019).
\newblock Understanding traffic capacity of urban networks.
\newblock {\em Scientific reports}, 9(1):1--10.

\bibitem[Lu et~al., 2020]{lu2020impact}
Lu, Q., Tettamanti, T., H{\"o}rcher, D., and Varga, I. (2020).
\newblock The impact of autonomous vehicles on urban traffic network capacity:
  an experimental analysis by microscopic traffic simulation.
\newblock {\em Transportation Letters}, 12(8):540--549.

\bibitem[Mariotte et~al., 2017]{mariotte2017macroscopic}
Mariotte, G., Leclercq, L., and Laval, J.~A. (2017).
\newblock Macroscopic urban dynamics: Analytical and numerical comparisons of
  existing models.
\newblock {\em Transportation Research Part B: Methodological}, 101:245--267.

\bibitem[Moshahedi and Kattan, 2022]{moshahedi2022macroscopic}
Moshahedi, N. and Kattan, L. (2022).
\newblock A macroscopic dynamic network loading model using variational theory
  in a connected and autonomous vehicle environment.
\newblock {\em Transportation research part C: emerging technologies},
  145:103911.

\bibitem[Ni and Cassidy, 2020]{ni2020city}
Ni, W. and Cassidy, M. (2020).
\newblock City-wide traffic control: modeling impacts of cordon queues.
\newblock {\em Transportation research part C: emerging technologies},
  113:164--175.

\bibitem[Ortigosa et~al., 2015]{ortigosa2015study}
Ortigosa, J., Menendez, M., and Tapia, H. (2015).
\newblock Study on the number and location of measurement points for an mfd
  perimeter control scheme: a case study of zurich.
\newblock {\em EURO Journal on Transportation and Logistics}, 3(3):245--266.

\bibitem[Ramezani et~al., 2015]{ramezani2015dynamics}
Ramezani, M., Haddad, J., and Geroliminis, N. (2015).
\newblock Dynamics of heterogeneity in urban networks: aggregated traffic
  modeling and hierarchical control.
\newblock {\em Transportation Research Part B: Methodological}, 74:1--19.

\bibitem[Simoni et~al., 2015]{simoni2015marginal}
Simoni, M.~D., Pel, A.~J., Waraich, R.~A., and Hoogendoorn, S.~P. (2015).
\newblock Marginal cost congestion pricing based on the network fundamental
  diagram.
\newblock {\em Transportation Research Part C: Emerging Technologies},
  56:221--238.

\bibitem[Sirmatel and Geroliminis, 2017]{sirmatel2017economic}
Sirmatel, I.~I. and Geroliminis, N. (2017).
\newblock Economic model predictive control of large-scale urban road networks
  via perimeter control and regional route guidance.
\newblock {\em IEEE Transactions on Intelligent Transportation Systems},
  19(4):1112--1121.

\bibitem[Sirmatel et~al., 2021]{sirmatel2021modeling}
Sirmatel, I.~I., Tsitsokas, D., Kouvelas, A., and Geroliminis, N. (2021).
\newblock Modeling, estimation, and control in large-scale urban road networks
  with remaining travel distance dynamics.
\newblock {\em Transportation Research Part C: Emerging Technologies},
  128:103157.

\bibitem[Small and Chu, 2003]{small2003hypercongestion}
Small, K.~A. and Chu, X. (2003).
\newblock Hypercongestion.
\newblock {\em Journal of Transport Economics and Policy (JTEP)},
  37(3):319--352.

\bibitem[Steck et~al., 2018]{steck2018autonomous}
Steck, F., Kolarova, V., Bahamonde-Birke, F., Trommer, S., and Lenz, B. (2018).
\newblock How autonomous driving may affect the value of travel time savings
  for commuting.
\newblock {\em Transportation research record}, 2672(46):11--20.

\bibitem[Takayama, 2020]{takayama2020gains}
Takayama, Y. (2020).
\newblock Who gains and who loses from congestion pricing in a monocentric city
  with a bottleneck?
\newblock {\em Economics of Transportation}, 24:100189.

\bibitem[Takayama and Kuwahara, 2017]{takayama2017bottleneck}
Takayama, Y. and Kuwahara, M. (2017).
\newblock Bottleneck congestion and residential location of heterogeneous
  commuters.
\newblock {\em Journal of Urban Economics}, 100:65--79.

\bibitem[Takayama and Kuwahara, 2020]{takayama2020scheduling}
Takayama, Y. and Kuwahara, M. (2020).
\newblock Scheduling preferences, parking competition, and bottleneck
  congestion: A model of trip timing and parking location choices by
  heterogeneous commuters.
\newblock {\em Transportation Research Part C: Emerging Technologies},
  117:102677.

\bibitem[Tsekeris and Geroliminis, 2013]{tsekeris2013city}
Tsekeris, T. and Geroliminis, N. (2013).
\newblock City size, network structure and traffic congestion.
\newblock {\em Journal of Urban Economics}, 76:1--14.

\bibitem[van~den Berg and Verhoef, 2016]{van2016autonomous}
van~den Berg, V.~A. and Verhoef, E.~T. (2016).
\newblock Autonomous cars and dynamic bottleneck congestion: The effects on
  capacity, value of time and preference heterogeneity.
\newblock {\em Transportation Research Part B: Methodological}, 94:43--60.

\bibitem[Vickrey, 2020]{vickrey2020congestion}
Vickrey, W. (2020).
\newblock Congestion in midtown manhattan in relation to marginal cost pricing.
\newblock {\em Economics of Transportation}, 21:100152.

\bibitem[Wu and Li, 2023]{wu2023managing}
Wu, S. and Li, Z.-C. (2023).
\newblock Managing a bi-modal bottleneck system with manned and autonomous
  vehicles: Incorporating the effects of in-vehicle activity utilities.
\newblock {\em Transportation Research Part C: Emerging Technologies},
  152:104179.

\bibitem[Xu et~al., 2018]{xu2018pareto}
Xu, S.-X., Liu, R., Liu, T.-L., and Huang, H.-J. (2018).
\newblock Pareto-improving policies for an idealized two-zone city served by
  two congestible modes.
\newblock {\em Transportation Research Part B: Methodological}, 117:876--891.

\bibitem[Yildirimoglu et~al., 2015]{yildirimoglu2015equilibrium}
Yildirimoglu, M., Ramezani, M., and Geroliminis, N. (2015).
\newblock Equilibrium analysis and route guidance in large-scale networks with
  mfd dynamics.
\newblock {\em Transportation Research Part C}, 59:404--420.

\bibitem[Yildirimoglu et~al., 2018]{yildirimoglu2018hierarchical}
Yildirimoglu, M., Sirmatel, I.~I., and Geroliminis, N. (2018).
\newblock Hierarchical control of heterogeneous large-scale urban road networks
  via path assignment and regional route guidance.
\newblock {\em Transportation Research Part B: Methodological}, 118:106--123.

\bibitem[Zheng et~al., 2017]{zheng2017macroscopic}
Zheng, N., Dantsuji, T., Wang, P., and Geroliminis, N. (2017).
\newblock Macroscopic approach for optimizing road space allocation of bus
  lanes in multimodal urban networks through simulation analysis.
\newblock {\em Transportation Research Record}, 2651(1):42--51.

\bibitem[Zheng and Geroliminis, 2013]{zheng2013distribution}
Zheng, N. and Geroliminis, N. (2013).
\newblock On the distribution of urban road space for multimodal congested
  networks.
\newblock {\em Transportation Research Part B: Methodological}, 57(C):326--341.

\bibitem[Zheng et~al., 2012]{zheng2012dynamic}
Zheng, N., Waraich, R.~A., Axhausen, K.~W., and Geroliminis, N. (2012).
\newblock A dynamic cordon pricing scheme combining the macroscopic fundamental
  diagram and an agent-based traffic model.
\newblock {\em Transportation Research Part A: Policy and Practice},
  46(8):1291--1303.

\bibitem[Zhong et~al., 2020]{zhong2020will}
Zhong, H., Li, W., Burris, M.~W., Talebpour, A., and Sinha, K.~C. (2020).
\newblock Will autonomous vehicles change auto commuters’ value of travel
  time?
\newblock {\em Transportation Research Part D: Transport and Environment},
  83:102303.

\bibitem[Zhou and Gayah, 2023]{zhou2023scalable}
Zhou, D. and Gayah, V.~V. (2023).
\newblock Scalable multi-region perimeter metering control for urban networks:
  A multi-agent deep reinforcement learning approach.
\newblock {\em Transportation Research Part C: Emerging Technologies},
  148:104033.

\end{thebibliography}
\end{document}